\documentclass[11pt,a4paper]{article}
\usepackage{amsfonts}
\usepackage[dvips]{graphicx}
\usepackage{amsmath}
\usepackage{mathrsfs}
\usepackage{makeidx}
\usepackage{xypic}
\usepackage[nottoc]{tocbibind}
\usepackage{pstricks}
\usepackage{bbm}
\usepackage[arrow, matrix, curve]{xy}
\usepackage{amsthm}
\usepackage{amssymb}
\usepackage{epic}
\usepackage{eepic}
\usepackage{epsfig}
\usepackage{hyperref}
\usepackage{cite}

\xymatrixcolsep{0.5cm}
\xymatrixrowsep{0.5cm}
\newtheorem{theorem}{Theorem}[section]
\newtheorem{proposition}{Proposition}[section]
\newtheorem{lemma}{Lemma}[section]

\newtheorem{corollary}{Corollary}[section]

\newcommand{\beqa}{\begin{eqnarray}}
\newcommand{\eeqa}{\end{eqnarray}}

\numberwithin{equation}{section}

\begin{document}

\title{\textbf{Lattice Sine-Gordon model}}
\author{\vspace{1cm}\vspace{2cm}}

\begin{flushright}
LPENSL-TH-07-16
\end{flushright}

\bigskip

\bigskip

\begin{center}
\textbf{\LARGE Transfer matrix spectrum for cyclic representations of the
6-vertex reflection algebra I}

\vspace{50pt}
\end{center}

\begin{center}
{\large \textbf{J.~M.~Maillet}\footnote[1]{{Univ Lyon, Ens de Lyon,
Univ Claude Bernard Lyon 1, CNRS, Laboratoire de Physique, UMR 5672, F-69342
Lyon, France; maillet@ens-lyon.fr}},~~ \textbf{G. Niccoli}\footnote[2]{%
{Univ Lyon, Ens de Lyon, Univ Claude Bernard Lyon 1, CNRS,
Laboratoire de Physique, UMR 5672, F-69342 Lyon, France;
giuliano.niccoli@ens-lyon.fr}},~~ \textbf{B. Pezelier}\footnote[3]{{
Univ Lyon, Ens de Lyon, Univ Claude Bernard Lyon 1, CNRS, Laboratoire de
Physique, UMR 5672, F-69342 Lyon, France; baptiste.pezelier@ens-lyon.fr}} }
\end{center}

\begin{center}
\vspace{50pt} \today \vspace{50pt}
\end{center}

\begin{itemize}
\item[ ] \textbf{Abstract.}\thinspace \thinspace We study the transfer
matrix spectral problem for the cyclic representations of the trigonometric
6-vertex reflection algebra associated to the Bazanov-Stroganov Lax
operator. The results apply as well to the spectral analysis of the lattice
sine-Gordon model with integrable open boundary conditions. This spectral
analysis is developed by implementing the method of separation of variables
(SoV). The transfer matrix spectrum (both eigenvalues and eigenstates) is
completely characterized in terms of the set of solutions to a discrete
system of polynomial equations in a given class of functions. Moreover, we
prove an equivalent characterization as the set of solutions to a Baxter's
like T-Q functional equation and rewrite the transfer matrix eigenstates in
an algebraic Bethe ansatz form. In order to explain our method in a simple
case, the present paper is restricted to representations containing one
constraint on the boundary parameters and on the parameters of the
Bazanov-Stroganov Lax operator. In a next article, some more technical tools
(like Baxter's gauge transformations) will be introduced to extend our
approach to general integrable boundary conditions.
\end{itemize}

\newpage\tableofcontents
\newpage

\section{Introduction}

The study of quantum models with integrable open boundary conditions has
attracted a large research interest, e.g. see \cite{OpenCyGau71,OpenCyAlca87,OpenCySkly88,OpenCyChe84,OpenCyKulS91,OpenCyMazNR90,OpenCyMazN91,OpenCyPS,ADMFBMNR,OpenCy,OpenCydeVegaG-R93,deVegaR-1994,OpenCyGhosZ94, OpenCyJimKKKM95-1, OpenCyJimKKMW95-2, OpenCyKKMNST07,OpenCyDoi03, OpenCyFrapNR07, OpenCyRag2+1, OpenCyRag2+2, OpenCyBK05, OpenCyN02, OpenCyNR03, OpenCyMNS06, OpenCyGal08, OpenCyYNZ06, ADMFCaoYSW13, OpenCyDerKM03-2, OpenCyFSW08, OpenCyFGSW11, OpenCyFram+2, OpenCyFram+1, OpenCyGN12-open, OpenCyFalN14, OpenCyFalKN14, OpenCyKitMN14, OpenCyFanHSY96, OpenCyCao03, OpenCyZhang07, OpenCyACDFR05, OpenCyRag1+, OpenCyRag2+, OpenCy-Hc, OpenCy-H1a, OpenCy-H1b, OpenCy-H1c, OpenCy-H1d, OpenCy-ShiW97, OpenCyAlcDHR94, OpenCy-Baj, OpenCydeGierE05,Doikou-2006}
and references therein. These models are of physical interest as they can
describe both equilibrium and out of equilibrium physics; e.g. some
interesting applications concern the description of classical stochastic
relaxation processes, like ASEP \cite{OpenCyD98,OpenCySch00},\cite%
{OpenCy-ShiW97, OpenCyAlcDHR94, OpenCy-Baj, OpenCydeGierE05}, and quantum
transport properties in spin systems \cite{SirPA09,Pro11}.

In this paper we start the analysis of the class of open integrable quantum
models associated to cyclic representations \cite{RocheA-1989,JimboDMM1990,JimboDMM1991,Tarasov-1992,Tarasov-1993} of the 6-vertex reflection
algebra. The literature of these models is so far rather sparse with the
exception of some rather special representations and boundary conditions,
like the open XXZ chains at the roots of unity, that can be traced back to
these representations under some special constraints, and for which some
results are known in the framework of algebraic Bethe ansatz (ABA) \cite%
{OpenCySF78,OpenCyFST80,Doikou-2006}, fusion of transfer matrices and truncation
identities \cite{OpenCyKRS, OpenCyKR, OpenCyRe83-1, OpenCyBR89,
OpenCyMazNR90,OpenCyMazN91}.

In order to study the general representations and boundary conditions we
have to go beyond traditional methods \cite{OpenCySF78, OpenCyFST80,
OpenCyKRS, OpenCyKR, OpenCyRe83-1, OpenCyBR89, OpenCyFT79, OpenCyS79,
OpenCyKS79, OpenCyF80, OpenCyS82, OpenCyF82, OpenCyF95, OpenCyJ90,
OpenCyKS82, OpenCySh85, OpenCyTh81, OpenCyIK82, OpenCyBe31, OpenCyBaxBook,
OpenCyLM66, OpenCyBBOY95} which do not apply for these general settings.
This is done by developing the Sklyanin's SoV method \cite{OpenCySk1,
Skl1985, OpenCySk2, OpenCySk3} for this class of models, a method that has
the advantage to lead (mainly by construction) to the complete
characterization of the spectrum (eigenvalues and eigenvectors) and has
proven to be applicable for a large variety of integrable quantum models 
\cite{OpenCyDerKM03-2, OpenCyFSW08, OpenCyFGSW11, OpenCyFram+2,
OpenCyFram+1, OpenCyGN12-open, OpenCyFalN14, OpenCyFalKN14, OpenCyKitMN14,
OpenCyBBS96, OpenCySm98, DerKM03, OpenCyBT06, OpenCyGIPS06, OpenCyGIPST07,
OpenCyNT-10, OpenCyN-10, OpenCyGN12, OpenCyGMN12-SG, OpenCyNWF09,
OpenCyN12-0, OpenCyN12-1, OpenCyNicT15-2, OpenCyNicT15}, where traditional
methods fail. Moreover, the SoV approach has the advantage to allow also for
the study of the dynamics of the models as it leads to universal determinant
formulae for matrix elements of local operators on transfer matrix
eigenstates as shown for different classes of models, first in \cite%
{OpenCyGMN12-SG}, and then in many other cases in \cite%
{OpenCyBBS96,OpenCySm98,OpenCyGN12-open,OpenCyN12-0,OpenCyN12-1,OpenCyLevNT15}%
. Moreover, the analysis developed in \cite{OpenCyKitMNT15, KitMNT16} makes
it possible to compute homogeneous and then thermodynamic limit of these
matrix elements opening the way to the computation of the corresponding
correlation functions.

Let us recall that in \cite{OpenCySkly88}, Sklyanin has shown how to
construct classes of quantum models with integrable boundaries in the
framework of the so-called Quantum Inverse Scattering Method \cite%
{OpenCySF78, OpenCyFST80, OpenCyKRS, OpenCyKR, OpenCyRe83-1, OpenCyBR89,
OpenCyFT79, OpenCyS79, OpenCyKS79, OpenCyF80, OpenCyS82, OpenCyF82,
OpenCyF95, OpenCyJ90, OpenCyKS82, OpenCySh85, OpenCyTh81, OpenCyIK82},
constructing in particular associated families of commuting transfer
matrices (conserved charges of the model). In fact, Sklyanin's construction
allows to use solutions of the Yang-Baxter equation to generate new
solutions of the reflection equation \cite{OpenCyChe84} once a scalar
solution of this last equation is known. Then as a consequence of the
reflection equation these new solutions generate commuting transfer matrices 
\cite{OpenCySkly88}.

Sklyanin has used the 6-vertex case and the associated XXZ spin 1/2 quantum
chains \cite{OpenCyBe31, OpenCyH28, OpenCyYY661, OpenCyGa83} to develop an
explicit example of this construction. However as pointed out in \cite%
{OpenCySkly88} similar constructions applies also to other 6-vertex cases
like the non-linear Schr\"{o}dinger and the Toda chains as well as for
models associated to the 8-vertex case like the XYZ spin 1/2 quantum chains.
Further integrable quantum models with open boundary conditions have been
presented following the Sklyanin's construction. Interesting examples are
the higher spin open quantum chains \cite{OpenCyDoi03, OpenCyFrapNR07}, the
higher rank open quantum spin chains \cite{OpenCyACDFR05, OpenCyRag1+,
OpenCyRag2+} and the Hubbard model \cite{OpenCy-Sha87, OpenCy-Ha, OpenCy-H1,
OpenCy-MarR98} with integrable open boundaries \cite{OpenCy-Hc, OpenCy-H1a,
OpenCy-H1b, OpenCy-H1c, OpenCy-H1d, OpenCy-ShiW97}. In fact, this mainly
results in the possibility to associate to any closed integrable quantum
model (characterized by a solution of the Yang-Baxter equation) new open
integrable quantum models (characterized by the associated Sklyanin's
solutions of the reflection equation).

Here, we characterize the spectral problem and we start the analysis of the
dynamical problem (by determining scalar products of separate states) for
the class of models in the Sklyanin's construction associated to general
scalar solutions of the 6-vertex reflection equation \cite%
{OpenCydeVegaG-R93,OpenCyGhosZ94} and the general Bazhanov-Stroganov cyclic
solution of the 6-vertex Yang-Baxter algebra \cite{OpenCyBS90}.

It may be instructive to recall some main literature on open integrable
quantum chains, even for different representations w.r.t. those studied
here. Indeed, this allows us to point out the difficulties that arise in
their analysis and that we have also encountered for the models studied in
this paper and to give the motivations for the approach that we have followed to overcome
them.

The spectrum of the open XXZ spin 1/2 quantum chain with parallel z-oriented
magnetic fields on the boundaries has been characterized in \cite%
{OpenCySkly88}, in the framework of the algebraic Bethe ansatz. While its
dynamics has been studied by the exact computation of correlation functions
first in \cite{OpenCyJimKKKM95-1,OpenCyJimKKMW95-2} and then in \cite%
{OpenCyKKMNST07}, there generalizing in the ABA framework the method
established in \cite{OpenCyKitMT99,OpenCyMaiT00} for the periodic chains.
These open quantum spin chains with z-oriented boundary magnetic fields
correspond in the Sklyanin's construction to the diagonal scalar solution of
the reflection equation. However, the most general scalar solution of the
6-vertex/8-vertex reflection equation is non-diagonal \cite%
{OpenCydeVegaG-R93, OpenCyGhosZ94} producing unparalleled and not z-oriented
boundaries magnetic fields. Under this general setting the analysis of the
spectrum and dynamics has shown to be much more involved. The ABA method
cannot be directly applied to these open chains with general boundary. In
fact, it was first understood in \cite{OpenCyFanHSY96}, for the XYZ spin 1/2
open chain, that the use of the Baxter's gauge transformations \cite%
{OpenCyBa72-1} allows to generalize the ABA method limitedly to non-diagonal
boundary matrices which satisfy one special constraint. After that, the same
approach, based now on the trigonometric version of the Baxter's gauge
transformations, was used in \cite{OpenCyCao03,OpenCyZhang07} to describe
the XXZ spectrum by ABA under similar constraints. Let us comment that for
the XXZ spectrum the same constraint was derived independently by a pure
functional method based on the use of the fusion of transfer matrices and
truncations identities for the roots of unit case in \cite{OpenCyN02}-\cite%
{OpenCyMNS06}. This has given access to the study of the spectrum leaving
however the study of the dynamics for these open models still unsolved.

Results on the spectrum for the most general unconstrained boundary
conditions have been achieved only more recently and they have required the
introduction of methods different from the ABA. Pure eigenvalue analysis has
been implemented in \cite{OpenCyGal08} by a functional method leading to
nested Bethe ansatz type equations similar to those previously introduced in 
\cite{OpenCyYNZ06}. Moreover, an ansatz for polynomial T-Q functional
equations with an inhomogeneous term has been recently argued in \cite%
{ADMFCaoYSW13}. Eigenstate construction has been first considered under
these general boundary in \cite{OpenCyBK05} by the q-Onsager algebra
formalism. A different approach, based on the generalization of the
Sklyanin's separation of variables (SoV) method to the reflection algebra
framework, has then lead to the complete eigenvalues and eigenstates
characterization \cite%
{OpenCyFSW08,OpenCyFGSW11,OpenCyGN12-open,OpenCyFalKN14,OpenCyFalN14,OpenCyKitMN14}%
, proving its equivalence to an inhomogeneous TQ functional equation \cite%
{OpenCyKitMN14}, and also giving access to first computations of matrix
elements of local operators \cite{OpenCyGN12-open} in the eigenstates basis.

The aim of the paper is to generalize this type of results for the 6-vertex
cyclic representations. Here we solve this problem in the case of one
triangular and one general boundary matrix, so that our current results
define also the setup for the solution of the most general boundary case as
well as the paper \cite{OpenCyGN12-open} has introduced the tools to solve
the case with the most general boundary in \cite{OpenCyFalKN14}.

The paper is organized as it follows. In section 2, we recall the cyclic
representations of the 6-vertex Yang-Baxter algebra associated to the 
Bazhanov-Stroganov Lax-operator. In section 3, we define the associated
representations of the cyclic 6-vertex reflection algebra. In section 4, we
prove the diagonalizability of the generator $\mathcal{B}_{-}(\lambda )$ of
the reflection algebra generated by $\mathcal{U}_{-}(\lambda )$ for the most
general $K_{-}(\lambda )$ boundary matrix while we impose one constraint on
the parameters of the Bazhanov-Stroganov Lax-operator for any quantum site
to make easier the explicit construction of the $\mathcal{B}_{-}$%
-eigenstates basis. Moreover, we compute the scalar product for the
so-called separate states in the $\mathcal{B}_{-}$-eigenstates basis. In
section 5, we show that the $\mathcal{B}_{-}$-eigenstates basis is the
SoV-basis for the transfer matrix spectral problem associated to the most
general $K_{-}$-boundary matrix and upper triangular $K_{+}$-boundary matrix
and we solve in this SoV basis this spectral problem. In section 6, we show
that the SoV characterization of the transfer matrix spectrum is equivalent
to inhomogeneous Baxter's like TQ-functional equation with polynomial
Q-functions. We present four appendices, in the first one we extend the
proof of section 4 for the diagonalizability of the operator $\mathcal{B}%
_{-}(\lambda )$ to the case of general values of the boundary and bulk
parameters. The remaining three appendices deal with the reduction of our
representations to those associated to the chiral-Potts, the sine-Gordon and
the XXZ spin s chains at the 2s+1 roots of unit.


\section{Cyclic representations of the 6-vertex Yang-Baxter algebra}

In this section we recall the basics of the cyclic representations of the
6-vertex Yang-Baxter algebra associated to the 
Bazhanov-Stroganov Lax-operator. We consider the representations defined by the
tensor product of $\mathsf{N}$ local representations of the 6-vertex
Yang-Baxter algebra on the local Hilbert spaces $\mathcal{R}_{n}$. Each
local representation is defined as the representation of a local Weyl algebra%
\begin{equation}
u_{n}v_{m}=q^{\delta _{n,m}}v_{m}u_{n}\text{ \ \ }\forall n,m\in \{1,...,%
\mathsf{N}\},
\end{equation}%
associated to a root of unit $q$, where $u_{n}$ and $v_{n}$ are the Weyl
algebra generators on the Hilbert spaces $\mathcal{R}_{n}$. Here, we assume
that $u_{n}$ and $v_{n}$ are unitary operators and that it holds: 
\begin{equation}
u_{n}^{p}=v_{n}^{p}=1\text{ for }q=e^{-i\pi \beta ^{2}}\text{, with }\beta
^{2}=p^{\prime }/p\text{ with }p^{\prime }\text{ even and }p=2l+1\text{ odd.}
\end{equation}%
This type of representation can be defined on a $p$-dimensional linear space 
$\mathcal{R}_{n}$, imposing that the $v_{n}$ spectrum coincides with the $p$%
-roots of the unit: 
\begin{equation}
v_{n}|k,n\rangle =q^{k}|k,n\rangle \text{ \ }\forall (n,k)\in \{1,...,%
\mathsf{N}\}\times \{-l,...,l\}.
\end{equation}%
On $\mathcal{R}_{n}$ is defined a $p$-dimensional representation of the Weyl
algebra by setting: 
\begin{equation}
u_{n}|k,n\rangle =|k+1,n\rangle \text{ \ }\forall k\in \{-l,...,l\}
\end{equation}%
with the cyclicity condition: 
\begin{equation}
|k+p,n\rangle =|k,n\rangle .
\end{equation}%
$\mathcal{R}_{n}$ is also called the right local quantum space at the site $%
n $ of the chain. Let $\mathcal{L}_{n}$ be the dual space of $\mathcal{R}%
_{n} $ then we can define the following scalar products: 
\begin{equation}
\langle k,n|k^{\prime },n\rangle =((\langle k,n|)^{\dagger },|k^{\prime
},n\rangle )\equiv \delta _{k,k^{\prime }},
\end{equation}%
for any $k,k^{\prime }\in \{-l,...,l\}$.

The local generators of the cyclic 6-vertex Yang-Baxter algebra can be now
defined as the elements of the following Bazhanov-Stroganov Lax operator: 
\begin{equation}
L_{a,n}(\lambda )\equiv \left( 
\begin{array}{cc}
\lambda \alpha _{n}v_{n}-\beta _{n}\lambda ^{-1}v_{n}^{-1} & u_{n}\left(
q^{-1/2}a_{n}v_{n}+q^{1/2}b_{n}v_{n}^{-1}\right) \\ 
u_{n}^{-1}\left( q^{1/2}c_{n}v_{n}+q^{-1/2}d_{n}v_{n}^{-1}\right) & \gamma
_{n}v_{n}/\lambda -\delta _{n}\lambda /v_{n}%
\end{array}%
\right) _{a}\in \text{End}(\mathbb{C}^{2}\otimes \mathcal{R}_{n}),
\end{equation}%
where $a$ denote the so-called auxiliary space $V_{a}\simeq \mathbb{C}^{2}$.
Indeed, under the condition%
\begin{equation}
\gamma _{n}=a_{n}c_{n}/\alpha _{n},\text{ \ \ \ \ }\delta
_{n}=b_{n}d_{n}/\beta _{n},
\end{equation}%
$L_{a,n}(\lambda )$ is a solution of the 6-vertex Yang-Baxter equation:%
\begin{equation}
R_{12}(\lambda /\mu )L_{1,n}(\lambda )L_{2,n}(\mu )=L_{2,n}(\mu
)L_{1,n}(\lambda )R_{12}(\lambda /\mu ),
\end{equation}%
w.r.t. the standard 6-vertex $R$-matrix:%
\begin{equation}
R_{ab}(\lambda )=\left( 
\begin{array}{cccc}
q\lambda -q^{-1}\lambda ^{-1} &  &  &  \\[-1mm] 
& \lambda -\lambda ^{-1} & q-q^{-1} &  \\[-1mm] 
& q-q^{-1} & \lambda -\lambda ^{-1} &  \\[-1mm] 
&  &  & q\lambda -q^{-1}\lambda ^{-1}%
\end{array}%
\right) \,,  \label{Rlsg}
\end{equation}%
where $a$ and $b$ denote two bidimensional spaces $V_{a},V_{b}\equiv \mathbb{%
C}^{2}$ and $R_{ab}(\lambda )$ is an endomorphism on their tensor product,
i.e. $R_{ab}(\lambda )\in $End$(\mathbb{C}^{2}\otimes \mathbb{C}^{2})$.
Then, the following monodromy matrix:%
\begin{equation}
M_{a}(\lambda )=\left( 
\begin{array}{cc}
A(\lambda ) & B(\lambda ) \\ 
C(\lambda ) & D(\lambda )%
\end{array}%
\right) _{a}\equiv L_{a,\mathsf{N}}(\lambda q^{-1/2})\cdots L_{a,1}(\lambda
q^{-1/2})\in \text{End}(\mathbb{C}^{2}\otimes \mathcal{H}),
\label{YB-monodromy}
\end{equation}%
is also a solution of the Yang-Baxter equation:%
\begin{equation}
R_{12}(\lambda /\mu )M_{1}(\lambda )M_{2}(\mu )=M_{2}(\mu )M_{1}(\lambda
)R_{12}(\lambda /\mu ),  \label{YBdef}
\end{equation}%
and its elements define a representation of the Yang-Baxter algebra on the
tensor product of the local representation spaces, i.e. $\mathcal{H}=\otimes
_{n=1}^{\mathsf{N}}\mathcal{R}_{n}$. Note that one can also consider cyclic
representations of the 6-vertex Yang-Baxter algebra associated to $q$, an
even root of unit, these have been studied in \cite{OpenCy-Au-YP08}.

\subsection{Bulk transfer matrix and quantum determinant}

The Yang-Baxter equation implies that the bulk transfer matrix $\tau
_{2}(\lambda )\equiv $tr$_{a}M_{a}(\lambda )$ defines a one parameter family
of commuting operators. Note that we have:%
\begin{equation}
\lbrack \tau _{2}(\lambda ),\Theta ]=0,\text{ \ \ where: \ \ }\Theta \equiv
\prod_{n=1}^{\mathsf{N}}v_{n}.
\end{equation}%
In \cite{OpenCyBS90,OpenCyBBP90} it was related to the analysis of the
chP-model \cite{OpenCyBa89, OpenCyAMcP0, OpenCyBaPauY, OpenCyauYMcPTY,
OpenCyMcPTS, OpenCyauYMcPT, OpenCyTarasovSChP, OpenCyBa04} and characterized
by SoV in \cite{OpenCyGIPS06, OpenCyGIPST07, OpenCyNT-10, OpenCyN-10,
OpenCyGN12, OpenCyGMN12-SG}. The Yang-Baxter equation also implies that the
so-called quantum determinant is a central element and it has the following
factorized form:%
\begin{eqnarray}
\mathrm{det_{q}}M_{a}(\lambda ) &\equiv &A(\lambda q^{1/2})D(\lambda
q^{-1/2})-B(\lambda q^{1/2})C(\lambda q^{-1/2}) \\
&=&D(\lambda q^{1/2})A(\lambda q^{-1/2})-C(\lambda q^{1/2})B(\lambda
q^{-1/2}) \\
&=&\prod_{n=1}^{\mathsf{N}}\mathrm{det_{q}}L_{a,n}(\lambda ),
\end{eqnarray}%
where the local quantum determinants read:%
\begin{eqnarray}
\mathrm{det_{q}}L_{a,n}(\lambda ) &\equiv &\left( L_{a,n}(\lambda )\right)
_{11}\left( L_{a,n}(\lambda q^{-1})\right) _{22}-\left( L_{a,n}(\lambda
)\right) _{12}\left( L_{a,n}(\lambda q^{-1})\right) _{21} \\
&=&\left( L_{a,n}(\lambda )\right) _{22}\left( L_{a,n}(\lambda
q^{-1})\right) _{11}-\left( L_{a,n}(\lambda )\right) _{21}\left(
L_{a,n}(\lambda q^{-1})\right) _{12} .
\end{eqnarray}%
They admit the following explicit form: 
\begin{eqnarray}
\mathrm{det_{q}}M_{a}(\lambda ) &=&\prod_{n=1}^{\mathsf{N}}k_{n}(\frac{%
\lambda }{\mu _{n,+}}-\frac{\mu _{n,+}}{\lambda })(\frac{\lambda }{\mu _{n,-}%
}-\frac{\mu _{n,-}}{\lambda }) \\
&=&(-q)^{\mathsf{N}}\prod_{n=1}^{\mathsf{N}}\frac{\beta _{n}a_{n}c_{n}}{%
\alpha _{n}}(\frac{1}{\lambda }+q^{-1}\frac{b_{n}\alpha _{n}}{a_{n}\beta _{n}%
}\lambda )(\frac{1}{\lambda }+q^{-1}\frac{d_{n}\alpha _{n}}{c_{n}\beta _{n}}%
\lambda ) \\
&=&a(\lambda )d(\lambda /q),
\end{eqnarray}%
where:%
\begin{eqnarray}
k_{n} &\equiv &\left( a_{n}b_{n}c_{n}d_{n}\right) ^{1/2},\text{ \ }\mu
_{n,h}\equiv \left\{ 
\begin{array}{c}
iq^{1/2}\left( a_{n}\beta _{n}/\alpha _{n}b_{n}\right) ^{1/2}\text{ \ \ }h=+,
\\ 
iq^{1/2}\left( c_{n}\beta _{n}/\alpha _{n}d_{n}\right) ^{1/2}\text{ \ \ }h=-.%
\end{array}%
\right. \\
a(\lambda ) &\equiv &a_{0}\prod_{n=1}^{\mathsf{N}}(\frac{\beta _{n}}{\lambda 
}+q^{-1}\frac{b_{n}\alpha _{n}}{a_{n}}\lambda ),\text{ \ }d(\lambda )\equiv 
\frac{(-1)^{\mathsf{N}}}{a_{0}}\prod_{n=1}^{\mathsf{N}}\frac{a_{n}c_{n}}{%
\alpha _{n}}(\frac{1}{\lambda }+q\frac{d_{n}\alpha _{n}}{c_{n}\beta _{n}}%
\lambda ),
\end{eqnarray}%
and $a_{0}$ is a free non zero parameter.

\section{Cyclic representations of the 6-vertex reflection algebra}

In this section we define the most general cyclic representations of the
6-vertex reflection algebra associated to the 
Bazhanov-Stroganov Lax-operator. This is done following the general procedure
introduced by Sklyanin \cite{OpenCySkly88} which allows to associate to any
solution $M_{a}(\lambda )\in $End$(\mathbb{C}^{2}\otimes \mathcal{H})$ of
the 6-vertex Yang-Baxter equation a solution $\mathcal{U}_{a}(\lambda )\in $%
End$(\mathbb{C}^{2}\otimes \mathcal{H})$ of the 6-vertex reflection equation:%
\begin{equation}
R_{12}(\lambda /\mu )\,\mathcal{U}_{1}(\lambda )\,R_{12}(\lambda \mu /q)\,%
\mathcal{U}_{2}(\mu )=\mathcal{U}_{2}(\mu )\,R_{12}(\lambda \mu /q)\,%
\mathcal{U}_{1}(\lambda )\,R_{12}(\lambda /\mu ).  \label{bYB}
\end{equation}%
Here, we have defined:%
\begin{equation}
\mathcal{U}_{a}(\lambda )=M_{a}(\lambda )K_{a}(\lambda )\hat{M}_{a}(\lambda
),
\end{equation}%
where $K_{a}(\lambda ;\zeta ,\kappa ,\tau )$ is the most general scalar
(boundary matrix) solution of the 6-vertex reflection equation \cite{OpenCydeVegaG-R93,deVegaR-1994}:%
\begin{equation}
K_{a}(\lambda ;\zeta ,\kappa ,\tau )=\frac{1}{\zeta -\frac{1}{\zeta }}\left( 
\begin{array}{cc}
\frac{\lambda \zeta }{q^{1/2}}-\frac{q^{1/2}}{\lambda \zeta } & \kappa
e^{\tau }\left( \frac{\lambda ^{2}}{q}-\frac{q}{\lambda ^{2}}\right) \\ 
\kappa e^{-\tau }\left( \frac{\lambda ^{2}}{q}-\frac{q}{\lambda ^{2}}\right)
& \frac{q^{1/2}\zeta }{\lambda }-\frac{\lambda }{\zeta q^{1/2}}%
\end{array}%
\right) _{a},
\end{equation}%
and we have defined:%
\begin{equation}
\hat{M}_{a}(\lambda )=(-1)^{\mathsf{N}}\,\sigma
_{a}^{y}\,M_{a}^{t_{a}}(1/\lambda )\,\sigma _{a}^{y}.  \label{Inverse-M}
\end{equation}%
Using this correspondence, the most general cyclic representations of the
6-vertex Yang-Baxter algebra, associated to the bulk monodromy matrix $%
\left( \ref{YB-monodromy}\right) $, define the most general cyclic
representations of the 6-vertex reflection algebra, corresponding to the
boundary monodromy matrices:%
\begin{eqnarray}
\mathcal{U}_{a,-}(\lambda ) &=&M_{a}(\lambda )K_{a,-}(\lambda )\hat{M}%
_{a}(\lambda )=\left( 
\begin{array}{cc}
\mathcal{A}_{-}(\lambda ) & \mathcal{B}_{-}(\lambda ) \\ 
\mathcal{C}_{-}(\lambda ) & \mathcal{D}_{-}(\lambda )%
\end{array}%
\right) _{a}, \\
\mathcal{U}_{a,+}^{t_{a}}(\lambda ) &=&M_{a}^{t_{a}}(\lambda
)K_{a,+}^{t_{a}}(\lambda )\hat{M}_{a}^{t_{a}}(\lambda )=\left( 
\begin{array}{cc}
\mathcal{A}_{+}(\lambda ) & \mathcal{C}_{+}(\lambda ) \\ 
\mathcal{B}_{+}(\lambda ) & \mathcal{D}_{+}(\lambda )%
\end{array}%
\right) _{a},
\end{eqnarray}%
with $\mathcal{U}_{a,-}(\lambda )$ and $\mathcal{V}_{a,+}(\lambda )\equiv 
\mathcal{U}_{a,+}^{t_{a}}(-\lambda )$ both solution of the reflection
equation $\left( \ref{bYB}\right) $, where:%
\begin{equation}
K_{a,\pm }(\lambda )=K_{a}(\lambda q^{(1\pm 1)/2};\zeta _{\pm },\kappa _{\pm
},\tau _{\pm })=\left( 
\begin{array}{cc}
a_{\pm }\left( \lambda \right) & b_{\pm }\left( \lambda \right) \\ 
c_{\pm }\left( \lambda \right) & d_{\pm }\left( \lambda \right)%
\end{array}%
\right) _{a},
\end{equation}%
and $\zeta _{\pm },\delta _{\pm },\tau _{\pm }$ are arbitrary complex
parameters.

\subsection{Boundary transfer matrix and quantum determinant}

These boundary monodromy matrices define a one parameter family of commuting
transfer matrices:%
\begin{eqnarray}
\mathcal{T}(\lambda ) &\equiv &\text{tr}_{a}\{K_{a,+}(\lambda
)\,M_{a}(\lambda )\,K_{a,-}(\lambda )\hat{M}_{a}(\lambda )\}
\label{OpenCytransfer} \\
&=&\text{tr}_{a}\{K_{a,-}(\lambda )\mathcal{U}_{a,+}(\lambda )\}=\text{tr}%
_{a}\{K_{a,+}(\lambda )\mathcal{U}_{a,-}(\lambda )\} \\
&=&a_{+}(\lambda )\mathcal{A}_{-}(\lambda )+d_{+}(\lambda )\mathcal{D}%
_{-}(\lambda )+b_{+}(\lambda )\mathcal{C}_{-}(\lambda )+c_{+}(\lambda )%
\mathcal{B}_{-}(\lambda ).
\end{eqnarray}%
This statement follows by using the reflection equation as Sklyanin has
proven in \cite{OpenCySkly88}. The characterization of the spectrum
(eigenvalues and eigenstates) of this class of transfer matrices is the main
subject of this paper. In particular, we will restrict our attention to the
special boundary condition $b_{+}(\lambda )=0$, which can be analyzed by
implementing the SoV approach once is proven the diagonalizability of the $%
\mathcal{B}_{-}(\lambda )$ family of commuting operators. In order to
introduce this spectral analysis we start pointing out some important
properties satisfied by the generators of the reflection algebra $\mathcal{A}%
_{-}(\lambda ),$ $\mathcal{B}_{-}(\lambda ),$ $\mathcal{C}_{-}(\lambda )$
and $\mathcal{D}_{-}(\lambda )$.

Let us start with the following re-parametrization of the boundary
parameters \cite{OpenCyN02}:%
\begin{equation}
\left( \alpha _{-}-1/\alpha _{-}\right) \left( \beta _{-}+1/\beta
_{-}\right) \equiv \frac{\zeta _{-}-1/\zeta _{-}}{\kappa _{-}},\text{ \ }%
\left( \alpha _{-}+1/\alpha _{-}\right) \left( \beta _{-}-1/\beta
_{-}\right) \equiv \frac{\zeta _{-}+1/\zeta _{-}}{\kappa _{-}}.
\label{Def-alfa-beta}
\end{equation}%
Then we define the following functions:%
\begin{equation}
\mathsf{A}_{-}(\lambda )\equiv g_{-}(\lambda )a(\lambda
q^{-1/2})d(1/(q^{1/2}\lambda )),
\end{equation}%
where:%
\begin{equation}
g_{-}(\lambda )\equiv \frac{(\lambda \alpha _{-}/q^{1/2}-q^{1/2}/(\lambda
\alpha _{-}))(\lambda \beta _{-}/q^{1/2}+q^{1/2}/(\lambda \beta _{-}))}{%
\left( \alpha _{-}-1/\alpha _{-}\right) \left( \beta _{-}+1/\beta
_{-}\right) }.  \label{B-g_PM}
\end{equation}

\begin{proposition}
The following boundary quantum determinant:%
\begin{eqnarray}
\text{det}_{q}\mathcal{U}_{a,-}(\lambda ) &\equiv &(\left( \lambda /q\right)
^{2}-\left( q/\lambda \right) ^{2})[\mathcal{A}_{-}(\lambda q^{1/2})\mathcal{%
A}_{-}(q^{1/2}/\lambda )+\mathcal{B}_{-}(\lambda q^{1/2})\mathcal{C}%
_{-}(q^{1/2}/\lambda )]  \label{Bound-q-detU_1} \\
&=&(\left( \lambda /q\right) ^{2}-\left( q/\lambda \right) ^{2})[\mathcal{D}%
_{-}(\lambda q^{1/2})\mathcal{D}_{-}(q^{1/2}/\lambda )+\mathcal{C}%
_{-}(\lambda q^{1/2})\mathcal{B}_{-}(q^{1/2}/\lambda )],
\label{Bound-q-detU_2}
\end{eqnarray}%
is a central element in the reflection algebra, i.e.%
\begin{equation}
\lbrack \text{det}_{q}\mathcal{U}_{a,-}(\lambda ),\mathcal{U}_{a,-}(\mu )]=0,
\end{equation}%
and its explicit expression reads:%
\begin{equation}
\text{det}_{q}\mathcal{U}_{a,-}(\lambda )=(\lambda ^{2}/q^{2}-q^{2}/\lambda
^{2})\mathsf{A}_{-}(\lambda q^{1/2})\mathsf{A}_{-}(q^{1/2}/\lambda ).
\label{B-q-detU_-exp}
\end{equation}%
Moreover, the generators of the reflection algebra satisfy the following
properties:%
\begin{equation}
\mathcal{D}_{-}(\lambda )=\frac{(\lambda ^{2}/q-q/\lambda ^{2})}{(\lambda
^{2}-1/\lambda ^{2})}\mathcal{A}_{-}(\lambda^{-1} )+\frac{(q-1/q)}{(\lambda
^{2}-1/\lambda ^{2})}\mathcal{A}_{-}(\lambda ),  \label{Sym-A-D-}
\end{equation}%
and%
\begin{equation}
\mathcal{B}_{-}(\lambda^{-1} )=-\frac{(\lambda ^{2}q-1/\left( q\lambda
^{2}\right) )}{(\lambda ^{2}/q-q/\lambda ^{2})}\mathcal{B}_{-}(\lambda )%
\text{ },\text{ \ }\mathcal{C}_{-}(\lambda^{-1} )=-\frac{(\lambda
^{2}q-1/\left( q\lambda ^{2}\right) )}{(\lambda ^{2}/q-q/\lambda ^{2})}%
\mathcal{C}_{-}(\lambda ).  \label{Sym-B-C-}
\end{equation}
\end{proposition}

We omit the proof of this proposition as it can be derived repeating the
main steps of the original Sklyanin's paper, where similar statements were
proven for the case of spin 1/2 representations of the 6-vertex reflection
algebra. Let us introduce now the following notation:%
\begin{equation}
\mathcal{T}_{\setminus }(\lambda )\equiv a_{+}\left( \lambda \right) 
\mathcal{A}_{-}(\lambda )+d_{+}\left( \lambda \right) \mathcal{D}%
_{-}(\lambda ),  \label{B-T-Diag}
\end{equation}%
for the diagonal part of the transfer matrix, i.e. the one associated to the
diagonal elements of the matrix $\mathcal{U}_{a,-}(\lambda )$, and the
coefficients:%
\begin{eqnarray}
\mathsf{a}_{+}(\lambda ) &\equiv &\frac{(\lambda ^{2}q-1/\left( q\lambda
^{2}\right) )(\lambda \zeta _{+}/q^{1/2}-q^{1/2}/(\lambda \zeta _{+}))}{%
(\lambda ^{2}-1/\lambda ^{2})(\zeta _{+}-1/\zeta _{+})}, \\
\mathsf{d}_{+}(\lambda ) &\equiv &\frac{(\lambda ^{2}q-1/\left( q\lambda
^{2}\right) )(\zeta _{+}q^{1/2}/\lambda -\lambda /(q^{1/2}\zeta _{+}))}{%
(\lambda ^{2}-1/\lambda ^{2})(\zeta _{+}-1/\zeta _{+})}.
\end{eqnarray}%
then we can prove the following:

\begin{corollary}
The most general transfer matrix admits the following symmetries:%
\begin{equation}
\mathcal{T}(\lambda )=\mathcal{T}(1/\lambda ),\text{ \ \ }\mathcal{T}%
(-\lambda )=\mathcal{T}(\lambda ),  \label{symmetry-transfer}
\end{equation}%
and the diagonal part $\mathcal{T}_{\setminus }(\lambda )$ has the following
explicitly symmetric forms:%
\begin{eqnarray}
\mathcal{T}_{\setminus }(\lambda ) &\equiv &\mathsf{a}_{+}(\lambda )\mathcal{%
A}_{-}(\lambda )+\mathsf{a}_{+}(1/\lambda )\mathcal{A}_{-}(1/\lambda )
\label{T-diag-A} \\
&=&\mathsf{d}_{+}(\lambda )\mathcal{D}_{-}(\lambda )+\mathsf{d}%
_{+}(1/\lambda )\mathcal{D}_{-}(1/\lambda ).  \label{T-diag-D}
\end{eqnarray}
\end{corollary}

\section{SoV\ representation of cyclic 6-vertex reflection algebra}

In this section we construct the left and right basis which diagonalize the
one-parameter family of commuting operators $\mathcal{B}_{-}(\lambda )$
associated to the most general $K_{-}(\lambda )$ matrix. Here we impose one
constraint on the parameters of the representation at any quantum site:%
\begin{equation}
b_{n}^{p}+a_{n}^{p}=0,\text{ \ \ }\forall n\in \{1,...,\mathsf{N}\}.
\label{SuperCH-P}
\end{equation}%
This is done in order to make completely explicit the construction of this
basis; however, the proof of the diagonalizability of $\mathcal{B}%
_{-}(\lambda )$ can be done without these constraints and under completly
general values of the inner boundary matrix and of the bulk parameters and
it will be presented in appendix.

\subsection{Pseudo-vacuum states}

We implement the above constraints by imposing:%
\begin{equation}
b_{n}=-q^{2j_{n}-1}a_{n},  \label{Quasi-nilp}
\end{equation}%
where for any $n\in \{1,...,\mathsf{N}\}$ we have fixed $j_{n}\in
\{0,...,p-1\}$, then we have:%
\begin{equation}
\left\langle j_{n}-1,n\right\vert \left( L_{a,n}\right) _{12}=\text{\b{0}, \
\ \ \ }\left( L_{a,n}\right) _{12}\left\vert j_{n},n\right\rangle =\text{\b{0%
},}
\end{equation}%
as well as:%
\begin{equation}
\left\langle j_{n}-1,n\right\vert \left( L_{a,n}(\lambda )\right) _{11}=%
\text{a}_{n}(\lambda q^{j_{n}-1})\left\langle j_{n}-1,n\right\vert \text{, \
\ \ \ }\left\langle j_{n}-1,n\right\vert \left( L_{a,n}(\lambda )\right)
_{22}=\text{d}_{n}(\lambda q^{1-j_{n}})\left\langle j_{n}-1,n\right\vert 
\text{,}
\end{equation}%
and%
\begin{equation}
\left( L_{a,n}(\lambda )\right) _{11}\left\vert j_{n},n\right\rangle
=\left\vert j_{n},n\right\rangle \text{a}_{n}(\lambda q^{j_{n}})\text{, \ \
\ \ \ \ }\left( L_{a,n}(\lambda )\right) _{22}\left\vert
j_{n},n\right\rangle =\left\vert j_{n},n\right\rangle \text{d}_{n}(\lambda
q^{-j_{n}})\text{,}
\end{equation}%
where:%
\begin{equation}
\text{a}_{n}(\lambda )=\lambda \alpha _{n}-\beta _{n}/\lambda ,\text{ \ \ d}%
_{n}(\lambda )=\gamma _{n}/\lambda -\lambda \delta _{n},
\end{equation}%
which is of course compatible with the local quantum determinant at site $n$:%
\begin{eqnarray}
\left\langle j_{n}-1,n\right\vert \mathrm{det_{q}}L_{a,n}(\lambda )
&=&\left\langle j_{n}-1,n\right\vert \left[ \left( L_{a,n}(\lambda )\right)
_{11}\left( L_{a,n}(\lambda /q)\right) _{22}-\left( L_{a,n}\right)
_{12}\left( L_{a,n}\right) _{21}\right] \\
&=&\text{a}_{n}(\lambda q^{j_{n}-1})\text{d}_{n}(\lambda
q^{-j_{n}})\left\langle j_{n}-1,n\right\vert \\
\mathrm{det_{q}}L_{a,n}(\lambda )\left\vert j_{n},n\right\rangle &=&\left[
\left( L_{a,n}(\lambda )\right) _{22}\left( L_{a,n}(\lambda /q)\right)
_{11}-\left( L_{a,n}\right) _{21}\left( L_{a,n}\right) _{12}\right]
\left\vert j_{n},n\right\rangle \\
&=&\left\vert j_{n},n\right\rangle \text{a}_{n}(\lambda q^{j_{n}-1})\text{d}%
_{n}(\lambda q^{-j_{n}})\text{,}
\end{eqnarray}%
being:%
\begin{equation}
\text{a}_{n}(\lambda q^{j_{n}-1})\text{d}_{n}(\lambda q^{-j_{n}})=-q\frac{%
\beta _{n}a_{n}c_{n}}{\alpha _{n}}(\frac{1}{\lambda }-q^{2(j_{n}-1)}\frac{%
\alpha _{n}}{\beta _{n}}\lambda )(\frac{1}{\lambda }+q^{-1}\frac{d_{n}\alpha
_{n}}{c_{n}\beta _{n}}\lambda ).
\end{equation}%
Then we can define the following left and right ``reference states":%
\begin{equation}
\left\langle \Omega \right\vert =\otimes _{n=1}^{\mathsf{N}}\left\langle
j_{n}-1,n\right\vert ,\text{ \ }\left\vert \bar{\Omega}\right\rangle
=\otimes _{n=1}^{\mathsf{N}}\left\vert j_{n},n\right\rangle.
\end{equation}%
\label{Ref-state-Con-L/R} The following properties are satisfied:%
\begin{eqnarray}
\left\langle \Omega \right\vert A(\lambda q^{1/2}) &=&a(\lambda
)\left\langle \Omega \right\vert ,\text{ \ }\left\langle \Omega \right\vert
D(\lambda q^{1/2})=d(\lambda )\left\langle \Omega \right\vert ,\text{ \ }%
\left\langle \Omega \right\vert B(\lambda )=\text{\b{0}},\text{ \ }%
\left\langle \Omega \right\vert C(\lambda )\neq \text{\b{0},}
\label{Left-identities-ref} \\
A(\lambda q^{1/2})\left\vert \bar{\Omega}\right\rangle &=&\left\vert \bar{%
\Omega}\right\rangle a(\lambda q),\text{ }D(\lambda q^{1/2})\left\vert \bar{%
\Omega}\right\rangle =\left\vert \bar{\Omega}\right\rangle d(\lambda /q),%
\text{ }B(\lambda )\left\vert \bar{\Omega}\right\rangle =\text{\b{0}},\text{ 
}C(\lambda )\left\vert \bar{\Omega}\right\rangle \neq \text{\b{0},}
\label{Right-identities-ref}
\end{eqnarray}%
where it is simple to verify that as it should:%
\begin{equation}
a(\lambda )=\prod_{n=1}^{\mathsf{N}}\text{a}_{n}(\lambda q^{j_{n}-1}),\text{
\ \ }d(\lambda )=\prod_{n=1}^{\mathsf{N}}\text{d}_{n}(\lambda q^{1-j_{n}}),
\end{equation}%
once we have fixed the free parameter:%
\begin{equation*}
a_{0}=\left( -q\right) ^{\mathsf{N}}\prod_{n=1}^{\mathsf{N}}q^{-j_{n}}.
\end{equation*}%
Of course, the coefficients $a(\lambda )$ and $d(\lambda )$ as well as the
reference states depend on the choice of the $\mathsf{N}$-tuple $\left\{
j_{1},...,j_{\mathsf{N}}\right\} $ but for simplicity we do not write it
explicitly.

\subsection{Representation of the reflection algebra in $\mathcal{B}_{-}(%
\protect\lambda )$-eigenstates basis}

\label{SoV-Rep-Refle-Alge}

The left and right SoV-representations of the cyclic 6-vertex reflection
algebra are now defined by constructing the left and right $\mathcal{B}%
_{-}(\lambda )$-eigenstates basis and by determining in this new basis the
representation of the other generators of the algebra. In order to present
our results we need to introduce some notations. We define the following
functions parametrized by the \textsf{$N$}-tuples $\text{\textbf{\textit{h}}} \equiv
(h_{1},...,h_{\mathsf{N}})\in \{0,...,p-1\}^{\mathsf{N}}$:%
\begin{equation}
\text{\textsc{b}}_{\text{\textbf{\textit{h}}}}(\lambda )\equiv \kappa _{-}e^{\tau
_{-}}\frac{(\lambda ^{2}/q-q/\lambda ^{2})}{(\zeta _{-}-1/\zeta _{-})}a_{%
\text{\textbf{\textit{h}}}}(\lambda )a_{\text{\textbf{\textit{h}}}}(1/\lambda ),
\label{EigenValue-B_}
\end{equation}%
with%
\begin{equation}
a_{\text{\textbf{\textit{h}}}}(\lambda )\equiv \left( -1\right) ^{\mathsf{N}%
}\prod_{n=1}^{\mathsf{N}}\left( \alpha _{n}\beta _{n}\right) ^{1/2}(\frac{%
\lambda }{\xi _{n}^{\left( h_{n}\right) }}-\frac{\xi _{n}^{\left(
h_{n}\right) }}{\lambda }),
\end{equation}%
where:%
\begin{equation}
\xi _{n}^{\left( h\right) }=\mu _{n,+}q^{h+1/2},\text{ }\xi _{n+\mathsf{N}%
}^{(h)}\equiv \xi _{n}^{(h)}\text{ \ }\forall n\in \{1,...,\mathsf{N}\},%
\text{ }a_{\text{\textbf{0}}}(\lambda )=a(\lambda /q^{1/2}).
\end{equation}%
Moreover, next, we will need also the following notations:%
\begin{eqnarray}
\Lambda &=&(\lambda ^{2}+1/\lambda ^{2}),\text{ \ \ \ }X_{b}^{(h_{b})}=(%
\zeta _{b}^{(h_{b})})^{2}+1/(\zeta _{b}^{(h_{b})})^{2},\text{ }X=q+1/q \\
\zeta _{n}^{(h)} &=&\left( \xi _{n}^{(h)}\right) ^{\varphi _{n}}\text{\ \
for \ }h\in \{0,...,p-1\}\text{\ \ and\ \ }\forall n\in \{1,...,2\mathsf{N}%
\}, \\
\varphi _{a} &=&1-2\theta (a-\mathsf{N})\text{ \ \ with \ }\theta (x)=\{0%
\text{ for }x\leq 0,\text{ }1\text{ for }x>0\}.
\end{eqnarray}

\begin{theorem}[\protect\underline{Left $\mathcal{B}_{-}(\protect\lambda )$
SOV-representations}]
If $b_{-}\left( \lambda \right) \neq 0$ and it holds:%
\begin{equation}
\mu _{n,+}^{p}\neq \mu _{m,+}^{p}\text{ \ \ }\forall n\neq m\in \{1,...,%
\mathsf{N}\},  \label{E-SOV}
\end{equation}%
and%
\begin{equation}
\mu _{n,+}^{2p}\neq \pm 1,\text{ }\mu _{n,+}^{2}\neq q^{-2h}\alpha
_{-}^{2\epsilon },\text{\ }\mu _{n,+}^{2}\neq -q^{-2h}\beta _{-}^{2\epsilon
},\text{\ }\mu _{n,+}^{2}\neq q^{-2\epsilon -2h}\mu _{m,-}^{2\epsilon }
\label{condition-SoV-2}
\end{equation}%
for any $\epsilon =\pm 1,$ $h\in \{1,...,p-1\}$ and $n,m\in \{1,...,\mathsf{N%
}\}$, then the states:%
\begin{equation}
\langle h_{1},...,h_{\mathsf{N}}|\equiv \frac{1}{\text{\textsc{n}}}\langle
\Omega |\prod_{n=1}^{\mathsf{N}}\prod_{k_{n}=1}^{h_{n}}\frac{\mathcal{A}%
_{-}(1/\xi _{n}^{\left( k_{n}-1\right) })}{\mathsf{A}_{-}(1/\xi _{n}^{\left(
k_{n}-1\right) })},  \label{Left-B-eigenstates}
\end{equation}%
where $h_{n}\in \{0,...,p-1\}$ and \textsc{n} is a free normalization,
define a $\mathcal{B}_{-}(\lambda )$-eigenstates basis of $\mathcal{H}^{\ast
}$:%
\begin{equation}
\langle \text{\textbf{\textit{h}}}|\mathcal{B}_{-}(\lambda )=\text{\textsc{b}}_{\text{%
\textbf{\textit{h}}}}(\lambda )\langle \text{\textbf{\textit{h}}}|.  \label{B-eigen-value}
\end{equation}%
Here we have denoted $\langle $\textbf{\textit{h}}$|\equiv \langle h_{1},...,h_{%
\mathsf{N}}|$. The quantum determinant and the following left action on the
generic state $\langle $\textbf{\textit{h}}$|$:%
\begin{align}
\langle \text{\textbf{\textit{h}}}|\mathcal{A}_{-}(\lambda )& =\sum_{a=1}^{2\mathsf{N}%
}\frac{(\lambda ^{2}/q-q/\lambda ^{2})(\lambda \zeta _{a}^{(h_{a})}-1/\zeta
_{a}^{(h_{a})}\lambda )}{((\zeta _{a}^{(h_{a})})^{2}/q-q/(\zeta
_{a}^{(h_{a})})^{2})((\zeta _{a}^{(h_{a})})^{2}-1/(\zeta _{a}^{(h_{a})})^{2})%
}\prod_{\substack{ b=1  \\ b\neq a\text{ mod}\mathsf{N}}}^{\mathsf{N}}\frac{%
\Lambda -X_{b}^{(h_{b})}}{X_{a}^{(h_{a})}-X_{b}^{(h_{b})}}\mathsf{A}%
_{-}(\zeta _{a}^{(h_{a})})  \notag \\
& \times \langle \text{\textbf{\textit{h}}}|\text{T}_{a}^{-\varphi _{a}}+(-1)^{%
\mathsf{N}}\text{det}_{q}M(1)\frac{(\lambda /q^{1/2}+q^{1/2}/\lambda )}{2}%
\prod_{b=1}^{\mathsf{N}}\frac{\Lambda -X_{b}^{(h_{b})}}{X-X_{b}^{(h_{b})}}%
\langle \text{\textbf{\textit{h}}}|  \notag \\
& +(-1)^{\mathsf{N}}\frac{(\zeta _{-}+1/\zeta _{-})}{(\zeta _{-}-1/\zeta
_{-})}\text{det}_{q}M(i)\frac{(\lambda /q^{1/2}-q^{1/2}/\lambda )}{2}%
\prod_{b=1}^{\mathsf{N}}\frac{\Lambda -X_{b}^{(h_{b})}}{X+X_{b}^{(h_{b})}}%
\langle \text{\textbf{\textit{h}}}|,
\end{align}%
where:%
\begin{equation}
\langle h_{1},...,h_{a},...,h_{\mathsf{N}}|\text{T}_{a}^{\pm }=\langle
h_{1},...,h_{a}\pm 1,...,h_{\mathsf{N}}|,
\end{equation}%
completely determine the representation of the other generators of the
reflection algebra in the $\mathcal{B}_{-}(\lambda )$-eigenstates basis.
Indeed, the representation of $\mathcal{D}_{-}(\lambda )$ follows from the
identity (\ref{Sym-A-D-}) while $\mathcal{C}_{-}(\lambda )$ by the quantum
determinant.\smallskip
\end{theorem}

\begin{proof}
Let us write explicitly the decomposition of the reflection algebra
generator:%
\begin{equation}
\mathcal{B}_{-}(\lambda )=\left( -1\right) ^{\mathsf{N}}(-a_{-}(\lambda
)A(\lambda )B(1/\lambda )+b_{-}(\lambda )A(\lambda )A(1/\lambda
)-c_{-}(\lambda )B(\lambda )B(1/\lambda )+d_{-}(\lambda )B(\lambda
)A(1/\lambda )),
\end{equation}%
in terms of the generators of the Yang-Baxter algebra. Then, by using the
identities $\left( \ref{Left-identities-ref}\right) $ it follows that $%
\left\langle \Omega \right\vert $ is a $\mathcal{B}_{-}(\lambda )$%
-eigenstate with non-zero eigenvalue:%
\begin{equation}
\left\langle \Omega \right\vert \mathcal{B}_{-}(\lambda )\equiv \text{%
\textsc{b}}_{\text{\textbf{0}}}(\lambda )\left\langle \Omega \right\vert .
\end{equation}%
Now by using the reflection algebra commutation relations:%
\begin{eqnarray}
\mathcal{A}_{-}(\lambda _{2})\mathcal{B}_{-}(\lambda _{1}) &=&\frac{(\lambda
_{1}q/\lambda _{2}-\lambda _{2}/(\lambda _{1}q))(\lambda _{1}\lambda
_{2}/q-q/(\lambda _{1}\lambda _{2}))}{(\lambda _{1}/\lambda _{2}-\lambda
_{2}/\lambda _{1})(\lambda _{1}\lambda _{2}-1/(\lambda _{1}\lambda _{2}))}%
\mathcal{B}_{-}(\lambda _{1})\mathcal{A}_{-}(\lambda _{2})  \notag \\
&&+\frac{(\lambda _{1}^{2}/q-q/\lambda _{1}^{2})(q-1/q)}{(\lambda
_{2}/\lambda _{1}-\lambda _{1}/\lambda _{2})(\lambda _{1}^{2}-1/\lambda
_{1}^{2})}\mathcal{B}_{-}(\lambda _{2})\mathcal{A}_{-}(\lambda _{1})  \notag
\\
&&-\frac{q-1/q}{(\lambda _{1}^{2}-1/\lambda _{1}^{2})(\lambda _{1}\lambda
_{2}-1/(\lambda _{1}\lambda _{2}))}\mathcal{B}_{-}(\lambda _{2})\mathcal{%
\tilde{D}}_{-}(\lambda _{1})  \label{OpenCybYB-AB}
\end{eqnarray}%
we can follow step by step the proof given in \cite{OpenCyN12-0} to prove
the validity of $(\ref{B-eigen-value})$. The action of $\mathcal{A}%
_{-}(\zeta _{b}^{(h_{b})})$ for $b\in \{1,...,2\mathsf{N}\}$ follows from
the definition of the states $\langle \bf{h} |$, the reflection algebra
commutation relations $\left( \ref{OpenCybYB-AB}\right) $ and the quantum
determinant relations. Let us show now that the conditions $\left( \ref%
{E-SOV}\right) $ and $\left( \ref{condition-SoV-2}\right) $ imply that the
set of p$^{\mathsf{N}}$ states $\langle $\textbf{\textit{h}}$|$ is a $\mathcal{B}%
_{-}(\lambda )$-eigenstates basis of $\mathcal{H}^{\ast }$. As by condition $%
\left( \ref{E-SOV}\right) $ each such state is associated to a different
eigenvalue of $\mathcal{B}_{-}(\lambda )$ the only thing that we need to
prove to get their linear independence is that each such state is nonzero.
We know by construction that the state $\left\langle \Omega \right\vert $ is
nonzero so let us assume by induction that the same is true for the state $%
\langle $\textbf{\textit{h}}$^{(0)}|=\langle h_{1}^{(0)},...,h_{\mathsf{N}}^{(0)}|$
with $h_{j}^{(0)}\in \{0,...,p-2\}$ and let us show that $\langle $\text{\textbf{\textit{h}}}%
$_{j}^{(0)}|=\langle h_{1}^{(0)},...,h_{j}^{(0)}+1,...,h_{\mathsf{N}}^{(0)}|$
is nonzero. We have that:%
\begin{equation}
\langle \text{\textbf{{\textit{h}}}}_{j}^{(0)}|\mathcal{A}_{-}(\xi _{j}^{(h_{j}^{(0)}+1)})=%
\mathsf{A}_{-}(\xi _{j}^{(h_{j}^{(0)}+1)})\langle {\text{\textbf{\textit{h}}}}^{(0)}|\neq 
\text{\b{0} \ }j\in \{1,...,\mathsf{N}\}
\end{equation}%
so that $\langle \text{\textbf{{\textit{h}}}}_{j}^{(0)}|$ is nonzero. Using this we can prove
that all the states $\langle $\textbf{\textit{h}}$^{(1)}|=\langle
h_{1}^{(0)}+x_{1},...,h_{\mathsf{N}}^{(0)}+x_{\mathsf{N}}|$ with $x_{j}\in
\{0,1\}$ for any $j\in \{1,...,\mathsf{N}\}$ are nonzero, which just prove
the validity of the induction. Finally, by using the identities:%
\begin{equation}
\mathcal{U}_{-}(q^{1/2})=\left( -1\right) ^{\mathsf{N}}\text{det}_{q}M(1)%
\text{ }I_{0},\text{ \ \ }\mathcal{U}_{-}(iq^{1/2})=i(-1)^{\mathsf{N}+1}%
\frac{\zeta _{-}+1/\zeta _{-}}{\zeta _{-}-1/\zeta _{-}}\text{det}_{q}M(i)%
\text{ }\sigma _{0}^{z},  \label{OpenCyU-identities}
\end{equation}%
and remarking that $\mathcal{A}_{-}(\lambda )$ has the following functional
dependence w.r.t. $\lambda $:%
\begin{equation}
\mathcal{A}_{-}(\lambda )=\sum_{a=0}^{2\mathsf{N}+1}\lambda ^{\left( 2a-2%
\mathsf{N}+1\right) }\mathcal{A}_{-,a},
\end{equation}%
where $\mathcal{A}_{-,a}\in $End$(\mathcal{H})$ are some fixed operators, we
get our interpolation formula for its action on $\langle $\textbf{\textit{h}}$|$.
\end{proof}

Similarly, defining:%
\begin{equation}
\kappa _{a}^{\left( h\right) }=k(\zeta _{a}^{(h)}),\text{ \ for }h\in
\{0,...,p-1\},\text{ }a\in \{1,...,2\mathsf{N}\},
\end{equation}%
and the function:%
\begin{equation}
k(\lambda )=\left( \lambda ^{2}-1/\lambda ^{2}\right) /(\lambda
^{2}/q^{2}-q^{2}/\lambda ^{2}),
\end{equation}%
we have similar properties for the right representations:

\begin{theorem}[\protect\underline{Right $\mathcal{B}_{-}(\protect\lambda )$
SOV-representations}]
If $b_{-}\left( \lambda \right) \neq 0$ and $\left( \ref{E-SOV}\right) $-$%
\left( \ref{condition-SoV-2}\right) $ are satisfied, the states:%
\begin{equation}
|h_{1},...,h_{\mathsf{N}}\rangle \equiv \frac{1}{\text{\textsc{n}}}%
\prod_{n=1}^{\mathsf{N}}\prod_{k_{n}=h_{n}}^{p-2}\frac{\mathcal{D}_{-}(\xi
_{n}^{\left( k_{n}+1\right) })}{\kappa _{n}^{\left( k_{n}+1\right) }\mathsf{A%
}_{-}(1/\xi _{n}^{\left( k_{n}\right) })}|\bar{\Omega}\rangle ,
\label{OpenCyD-right-eigenstates}
\end{equation}%
with {\text{\textsc{n}}} the same coefficient as in %
\eqref{Left-B-eigenstates}, define a $\mathcal{B}_{-}(\lambda )$-eigenstates
basis of $\mathcal{H}$:%
\begin{equation}
\mathcal{B}_{-}(\lambda )|\text{\textbf{\textit{h}}}\rangle =|\text{\textbf{\textit{h}}}%
\rangle \text{\textsc{b}}_{\text{\textbf{\textit{h}}}}(\lambda ).
\label{OpenCyleft-B-eigen-cond}
\end{equation}%
On the generic state $|$\textbf{\textit{h}}$\rangle $, the action of the remaining
reflection algebra generators follows by:%
\begin{align}
\mathcal{D}_{-}(\lambda )|\text{\textbf{\textit{h}}}\rangle & =\sum_{a=1}^{2\mathsf{N}%
}\text{T}_{a}^{-\varphi _{a}}|\text{\textbf{\textit{h}}}\rangle \frac{(\lambda
^{2}/q-q/\lambda ^{2})(\lambda \zeta _{a}^{(h_{a})}-1/\zeta
_{a}^{(h_{a})}\lambda )}{((\zeta _{a}^{(h_{a})})^{2}/q-q/(\zeta
_{a}^{(h_{a})})^{2})((\zeta _{a}^{(h_{a})})^{2}-1/(\zeta _{a}^{(h_{a})})^{2})%
}\prod_{\substack{ b=1  \\ b\neq a\text{ mod}\mathsf{N}}}^{\mathsf{N}}\frac{%
\Lambda -X_{b}^{(h_{b})}}{X_{a}^{(h_{a})}-X_{b}^{(h_{b})}}  \notag \\
&\times\mathsf{D}_{-}(\zeta _{a}^{(h_{a})}) +|\text{\textbf{\textit{h}}}\rangle (-1)^{%
\mathsf{N}}\text{det}_{q}M(1)\frac{(\lambda /q^{1/2}+q^{1/2}/\lambda )}{2}%
\prod_{b=1}^{\mathsf{N}}\frac{\Lambda -X_{b}^{(h_{b})}}{X-X_{b}^{(h_{b})}} 
\notag \\
& +(-1)^{\mathsf{N}+1}|\text{\textbf{\textit{h}}}\rangle \frac{(\zeta _{-}+1/\zeta
_{-})}{(\zeta _{-}-1/\zeta _{-})}\text{det}_{q}M(i)\frac{(\lambda
/q^{1/2}-q^{1/2}/\lambda )}{2}\prod_{b=1}^{\mathsf{N}}\frac{\Lambda
-X_{b}^{(h_{b})}}{X+X_{b}^{(h_{b})}},  \label{SOV D-}
\end{align}%
where:%
\begin{equation}
\mathsf{D}_{-}(\lambda )=k(\lambda )\mathsf{A}_{-}(q/\lambda ),\text{\ \ T}%
_{a}^{\pm }|h_{1},...,h_{a},...,h_{\mathsf{N}}\rangle =|h_{1},...,h_{a}\pm
1,...,h_{\mathsf{N}}\rangle .
\end{equation}%
Indeed, the representation of $\mathcal{A}_{-}(\lambda )$ follows from the
identity (\ref{Sym-A-D-}) while $\mathcal{C}_{-}(\lambda )$ is given by the
quantum determinant.\smallskip
\end{theorem}

\begin{proof}
The proof is given along the same steps used in the previous theorem, we
just need to make the following remarks. First of all by using the
identities $\left( \ref{Right-identities-ref}\right) $ it follows that $%
\left\vert \bar{\Omega}\right\rangle $ is a $\mathcal{B}_{-}(\lambda )$%
-eigenstate with non-zero eigenvalue:%
\begin{equation}
\mathcal{B}_{-}(\lambda )\left\vert \bar{\Omega}\right\rangle \equiv
\left\vert \bar{\Omega}\right\rangle \text{\textsc{b}}_{\text{\textbf{p-1}}%
}(\lambda ).
\end{equation}%
Now all we need are the following reflection algebra commutation relations:%
\begin{align}
\mathcal{B}_{-}(\lambda _{1})\mathcal{D}_{-}(\lambda _{2})& =\frac{(\lambda
_{1}q/\lambda _{2}-\lambda _{2}/(\lambda _{1}q))(\lambda _{1}\lambda
_{2}/q-q/(\lambda _{1}\lambda _{2}))}{(\lambda _{1}/\lambda _{2}-\lambda
_{2}/\lambda _{1})(\lambda _{1}\lambda _{2}-1/(\lambda _{1}\lambda _{2}))}%
\mathcal{D}_{-}(\lambda _{2})\mathcal{B}_{-}(\lambda _{1})  \notag \\
& -\frac{(q-1/q)(\lambda _{1}\lambda _{2}/q-q/(\lambda _{1}\lambda _{2}))}{%
(\lambda _{1}/\lambda _{2}-\lambda _{2}/\lambda _{1})(\lambda _{1}\lambda
_{2}-1/(\lambda _{1}\lambda _{2}))}\mathcal{D}_{-}(\lambda _{1})\mathcal{B}%
_{-}(\lambda _{2})  \notag \\
& -\frac{q-1/q}{(\lambda _{1}\lambda _{2}-1/(\lambda _{1}\lambda _{2}))}%
\mathcal{A}_{-}(\lambda _{1})\mathcal{B}_{-}(\lambda _{2}).
\end{align}%
By using them, the definition of the states $|$\textbf{\textit{h}}$\rangle $ and the
quantum determinant, we get our interpolation formula for the right action
of $\mathcal{D}_{-}(\lambda )$ on $|$\textbf{\textit{h}}$\rangle $. Let us remark
that in fact, the chosen gauge for the coefficients of $\mathcal{D}%
_{-}(\lambda )$ is consistent with the quantum determinant condition as we
have:%
\begin{equation}
\mathsf{D}_{-}(\zeta _{a}^{(h)})=\kappa _{a}^{\left( h\right) }\mathsf{A}%
_{-}(q/\zeta _{a}^{(h)})
\end{equation}%
for any $h\in \{0,...,p-1\}$ and $a\in \{1,...,2\mathsf{N}\}$ and%
\begin{equation}
\frac{\text{det}_{q}\mathcal{U}_{-}(\xi _{a}^{(h+1/2)})}{(\,( \xi
_{a}^{(h+3/2)}) ^{2}-1/( \xi _{a}^{(h+3/2)}) ^{2}\,)}=\mathsf{D}_{-}(\xi
_{a}^{(h+1)})\mathsf{D}_{-}(1/\xi _{a}^{(h)})=\mathsf{A}_{-}(\xi
_{a}^{(h+1)})\mathsf{A}_{-}(1/\xi _{a}^{(h)}),
\end{equation}%
since,%
\begin{equation}
\kappa _{a}^{\left( h+1\right) }\kappa _{a+\mathsf{N}}^{\left( h\right) }=1
\end{equation}%
for any $h\in \{0,...,p-1\}$ and $a\in \{1,...,\mathsf{N}\}$.
\end{proof}

\subsection{Change of basis and SoV spectral decomposition of the identity}

In this section we present the main properties of the $p^{\mathsf{N}}\times
p^{\mathsf{N}}$ matrices $U^{(L)}$ and $U^{(R)}$ defining respectively the
change of basis from the original left and right basis, formed by the $v_{n}$%
-eigenstates basis:%
\begin{equation}
\underline{\langle \text{\textbf{\textit{h}}}|}\equiv \otimes _{n=1}^{\mathsf{N}%
}\langle h_{n},n|\text{ \ \ \ \ and \ \ \ }\underline{|\text{\textbf{\textit{h}}}%
\rangle }\equiv \otimes _{n=1}^{\mathsf{N}}|h_{n},n\rangle ,
\end{equation}%
to the left and right $\mathcal{B}_{-}$-eigenstates basis: 
\begin{equation}
\langle \text{\textbf{\textit{h}}}|=\underline{\langle \text{\textbf{\textit{h}}}|}%
U^{(L)}=\sum_{i=1}^{p^{\mathsf{N}}}U_{\varkappa \left( \text{\textbf{\textit{h}}}%
\right) ,i}^{(L)}\langle \varkappa ^{-1}\left( i\right) |\text{ \ \ and\ \ \ 
}|\text{\textbf{\textit{h}}}\rangle =U^{(R)}\underline{|\text{\textbf{\textit{h}}}\rangle }%
=\sum_{i=1}^{p^{\mathsf{N}}}U_{i,\varkappa \left( \text{\textbf{\textit{h}}}\right)
}^{(R)}|\varkappa ^{-1}\left( i\right) \rangle ,
\end{equation}%
where $\varkappa $ is the isomorphism between the sets $\{0,...,p-1\}^{%
\mathsf{N}}$ and $\{1,...,p^{\mathsf{N}}\}$ defined by: 
\begin{equation}
\varkappa : \text{\textbf{\textit{h}}}\in \{0,...,p-1\}^{\mathsf{N}}\rightarrow
\varkappa \left( \text{\textbf{\textit{h}}}\right) \equiv 1+\sum_{a=1}^{\mathsf{N}%
}p^{(a-1)}h_{a}\in \{1,...,p^{\mathsf{N}}\}.  \label{corrisp}
\end{equation}%
From the diagonalizability of $\mathcal{B}_{-}(\lambda )$ it follows that $%
U^{(L)}$ and $U^{(R)}$\ are invertible matrices for which it holds:%
\begin{equation}
U^{(L)}\mathcal{B}_{-}(\lambda )=\Delta _{\mathcal{B}_{-}}(\lambda )U^{(L)},%
\text{ \ \ }\mathcal{B}_{-}(\lambda )U^{(R)}=U^{(R)}\Delta _{\mathcal{B}%
_{-}}(\lambda ),
\end{equation}%
where $\Delta _{\mathcal{B}_{-}}(\lambda )$ is the $p^{\mathsf{N}}\times p^{%
\mathsf{N}}$ diagonal matrix defined by: 
\begin{equation}
\left( \Delta _{\mathcal{B}_{-}}(\lambda )\right) _{i,j}\equiv \delta _{i,j}%
\text{\textsc{b}}_{\varkappa ^{-1}\left( i\right) }(\lambda )\text{ \ }%
\forall i,j\in \{1,...,p^{\mathsf{N}}\}.
\end{equation}%
We can prove that it holds:

\begin{proposition}
The $p^{\mathsf{N}}\times p^{\mathsf{N}}$ matrix $M\equiv U^{(L)}U^{(R)}$ of scalar products of left and right $\mathcal{B}_{-}$-eigenstates  is
diagonal and it is characterized by the following diagonal entries:%
\begin{equation}
M_{\varkappa \left( \text{\textbf{\textit{h}}}\right) \varkappa \left( \text{\textbf{h%
}}\right) }=\langle \text{\textbf{\textit{h}}}|\text{\textbf{\textit{h}}}\rangle =\prod_{1\leq
b<a\leq \mathsf{N}}\frac{1}{X_{a}^{(h_{a})}-X_{b}^{(h_{b})}}.  \label{T2M_jj}
\end{equation}
\end{proposition}

\begin{proof}
Note that the action of a left $\mathcal{B}_{-}$-eigenstate on a right $%
\mathcal{B}_{-}$-eigenstate is zero, for two different $\mathcal{B}_{-}$%
-eigenvalues. This implies that the matrix $M$ is diagonal; then to compute
its diagonal elements we compute the matrix elements $\theta _{a}\equiv
\langle h_{1},...,h_{a},...,h_{\mathsf{N}}|\mathcal{A}_{-}(\xi _{a}^{\left(
h_{a}+1\right) })|h_{1},...,h_{a}+1,...,h_{\mathsf{N}}\rangle $, where $a\in
\{1,...,\mathsf{N}\}$. Using the left action of the operator $\mathcal{A}%
_{-}(\xi _{a}^{\left( h_{a}+1\right) })$ we get:%
\begin{align}
\theta _{a}& =\frac{(q-1/q)\mathsf{A}_{-}(1/\xi _{a}^{\left( h_{a}\right) })%
}{((\xi _{a}^{(h_{a})})^{2}-1/(\xi _{a}^{(h_{a})})^{2})}\prod_{\substack{ %
b=1  \\ b\neq a}}^{\mathsf{N}}\frac{X_{a}^{(h_{a}+1)}-X_{b}^{(h_{b})}}{%
X_{a}^{(h_{a})}-X_{b}^{(h_{b})}}  \notag \\
& \times \left. \langle h_{1},...,h_{a}+1,...,h_{\mathsf{N}%
}|h_{1},...,h_{a}+1,...,h_{\mathsf{N}}\rangle \right.
\end{align}%
while using the decomposition (\ref{Sym-A-D-}) and the fact that:%
\begin{equation}
\langle h_{1},...,h_{a},...,h_{\mathsf{N}}|\mathcal{D}_{-}(1/\xi
_{a}^{\left( h_{a}+1\right) })|h_{1},...,h_{a}+1,...,h_{\mathsf{N}}\rangle =0
\end{equation}%
it holds:%
\begin{equation}
\theta _{a}=\frac{k_{a}^{\left( h_{a}+1\right) }(q-1/q)\mathsf{A}_{-}(1/\xi
_{a}^{\left( h_{a}\right) })}{((\xi _{a}^{(h_{a}+1)})^{2}-1/(\xi
_{a}^{(h_{a}+1)})^{2})}\langle h_{1},...,h_{a},...,h_{\mathsf{N}%
}|h_{1},...,h_{a},...,h_{\mathsf{N}}\rangle ,
\end{equation}%
and so:%
\begin{equation}
\theta _{a}=\frac{(q-1/q)\mathsf{A}_{-}(1/\xi _{a}^{\left( h_{a}\right) })}{%
((\xi _{a}^{(h_{a})})^{2}-1/(\xi _{a}^{(h_{a})})^{2})}\langle
h_{1},...,h_{a},...,h_{\mathsf{N}}|h_{1},...,h_{a},...,h_{\mathsf{N}}\rangle
.
\end{equation}%
These results lead to the identity:%
\begin{equation}
\frac{\langle h_{1},...,h_{a}+1,...,h_{\mathsf{N}}|h_{1},...,h_{a}+1,...,h_{%
\mathsf{N}}\rangle }{\langle h_{1},...,h_{a},...,h_{\mathsf{N}%
}|h_{1},...,h_{a},...,h_{\mathsf{N}}\rangle }=\prod_{\substack{ b=1  \\ %
b\neq a}}^{\mathsf{N}}\frac{X_{a}^{(h_{a})}-X_{b}^{(h_{b})}}{%
X_{a}^{(h_{a}+1)}-X_{b}^{(h_{b})}},  \label{T2F1}
\end{equation}%
from which one can prove:%
\begin{equation}
\frac{\langle h_{1},...,h_{\mathsf{N}}|h_{1},...,h_{\mathsf{N}}\rangle }{%
\langle p-1,...,p-1|p-1,...,p-1\rangle }=\prod_{1\leq b<a\leq \mathsf{N}}%
\frac{X_{a}^{\left( p-1\right) }-X_{b}^{\left( p-1\right) }}{X_{a}^{\left(
h_{a}\right) }-X_{b}^{\left( h_{b}\right) }}.  \label{T2F2}
\end{equation}%
This proves the proposition being%
\begin{equation}
\langle p-1,...,p-1|p-1,...,p-1\rangle =\prod_{1\leq b<a\leq \mathsf{N}}%
\frac{1}{X_{a}^{\left( p-1\right) }-X_{b}^{\left( p-1\right) }},
\end{equation}%
using the following choice of the normalization:%
\begin{equation}
\text{\textsc{n}}=\left( \prod_{1\leq b<a\leq \mathsf{N}}\left(
X_{a}^{\left( p-1\right) }-X_{b}^{\left( p-1\right) }\right) \left\langle
\Omega \right\vert \prod_{n=1}^{\mathsf{N}}\prod_{k_{n}=0}^{p-2}\frac{%
\mathcal{A}_{-}(1/\xi _{n}^{\left( k_{n}\right) })}{\mathsf{A}_{-}(1/\xi
_{n}^{\left( k_{n}\right) })}|\bar{\Omega}\rangle \right) ^{1/2}.
\end{equation}
\end{proof}

The previous theorem implies the following spectral decomposition of the
identity $\mathbb{I}$ in the SoV basis:%
\begin{equation}
\mathbb{I}\equiv \sum_{h_{1},...,h_{\mathsf{N}}=0}^{p-1}\prod_{1\leq b<a\leq 
\mathsf{N}}(X_{a}^{(h_{a})}-X_{a}^{(h_{a})})|h_{1},...,h_{\mathsf{N}}\rangle
\langle h_{1},...,h_{\mathsf{N}}|.  \label{Decmp-Id}
\end{equation}

\subsection{Separate states and their scalar products}

Let us introduce a class of left and right states, the so-called separate
states, characterized by the following type of decompositions in the left
and right SoV-basis: 
\begin{align}
\langle \alpha |& =\sum_{h_{1},...,h_{\mathsf{N}}=0}^{p-1}\prod_{a=1}^{%
\mathsf{N}}\alpha _{a}^{(h_{a})}\prod_{1\leq b<a\leq \mathsf{N}%
}(X_{a}^{(h_{a})}-X_{b}^{(h_{b})})\langle h_{1},...,h_{\mathsf{N}}|,
\label{Separate-left-SoV} \\
|\beta \rangle & =\sum_{h_{1},...,h_{\mathsf{N}}=0}^{p-1}\prod_{a=1}^{%
\mathsf{N}}\beta _{a}^{(h_{a})}\prod_{1\leq b<a\leq \mathsf{N}%
}(X_{a}^{(h_{a})}-X_{b}^{(h_{b})})|h_{1},...,h_{\mathsf{N}}\rangle ,
\label{Separate-right-SoV}
\end{align}%
where the coefficients $\alpha _{a}^{(h_{a})}$ and $\beta _{a}^{(h_{a})}$
are arbitrary complex numbers, meaning that the coefficients of these
separate states have a factorized form in this basis. These separate states
are interesting at least for two reasons : they admit simple determinant
scalar products, as it will be shown in the next proposition, and the
eigenstates of the transfer matrix are special separate states, as we will
show in the next section.

\begin{proposition}
Let us take an arbitrary separate left state $\langle \alpha |$ (separate
covector)\ and an arbitrary separate right state $|\beta \rangle $ (separate
vector) then it holds:%
\begin{equation}
\langle \alpha |\beta \rangle =det_{\mathsf{N}}||\mathcal{M}_{a,b}^{\left(
\alpha ,\beta \right) }||\text{ \ \ with \ }\mathcal{M}_{a,b}^{\left( \alpha
,\beta \right) }\equiv \sum_{h=0}^{p-1}\alpha _{a}^{(h)}\beta
_{a}^{(h)}(X_{a}^{(h)})^{(b-1)}.  \label{T2-Sov-Sc-p1}
\end{equation}
\end{proposition}

\begin{proof}
The proof follows the same method as in \cite{OpenCyGMN12-SG}. The formula $%
\left( \ref{T2M_jj}\right) $ implies, using the representation of the states 
$\langle \alpha |$ and $|\beta \rangle $, the following:%
\begin{equation}
\langle \alpha |\beta \rangle =\sum_{h_{1},...,h_{\mathsf{N}%
}=0}^{p-1}V(X_{1}^{(h_{1})},...,X_{\mathsf{N}}^{(h_{\mathsf{N}%
})})\prod_{a=1}^{\mathsf{N}}\alpha _{a}^{(h_{a})}\beta _{a}^{(h_{a})},
\end{equation}%
where we have denoted by $V(x_{1},...,x_{\mathsf{N}})\equiv \prod_{1\leq
b<a\leq \mathsf{N}}(x_{a}-x_{b})$ the Vandermonde determinant. Finally,
using the multilinearity of the determinant we get our result.
\end{proof}

\section{$\mathcal{T}$-spectrum characterization in the SoV basis}

In this section we present the complete characterization of the spectrum of
the transfer matrix $\mathcal{T}(\lambda )$ associated to the cyclic
representations of the 6-vertex reflection algebra. We first present some
preliminary properties satisfied by all the eigenvalue functions of the
transfer matrix $\mathcal{T}(\lambda )$:

\begin{lemma}
Denote by $\Sigma _{\mathcal{T}}$ the transfer matrix spectrum, then any $%
\tau (\lambda )\in \Sigma _{\mathcal{T}}$ is an even function of $\lambda $\
invariant under the transformation $\lambda \rightarrow 1/\lambda $ which
admits the following interpolation formula:%
\begin{align}
\tau (\lambda )=& \sum_{a=1}^{\mathsf{N}}\frac{\Lambda ^{2}-X^{2}}{%
(X_{a}^{(0)})^{2}-X^{2}}\prod_{\substack{ b=1  \\ b\neq a}}^{\mathsf{N}}%
\frac{\Lambda -X_{b}^{(0)}}{X_{a}^{(0)}-X_{b}^{(0)}}\tau (\zeta
_{a}^{(0)})+(-1)^{\mathsf{N}}\frac{(\Lambda +X)}{2}\prod_{b=1}^{\mathsf{N}}%
\frac{\Lambda -X_{b}^{(0)}}{X-X_{b}^{(0)}}\text{det}_{q}M(1)  \notag \\
& -(-1)^{\mathsf{N}}\frac{(\Lambda -X)}{2}\prod_{b=1}^{\mathsf{N}}\frac{%
\Lambda -X_{b}^{(0)}}{X+X_{b}^{(0)}}\frac{(\zeta _{+}+1/\zeta _{+})}{(\zeta
_{+}-1/\zeta _{+})}\frac{(\zeta _{-}+1/\zeta _{-})}{(\zeta _{-}-1/\zeta _{-})%
}\text{det}_{q}M(i)  \notag \\
& +(\Lambda ^{2}-X^{2})\tau _{\infty }\prod_{b=1}^{\mathsf{N}}(\Lambda
-X_{b}^{(0)}),  \label{set-tau}
\end{align}%
where:%
\begin{equation}
\tau _{\infty }\equiv \frac{\kappa _{+}\kappa _{-}(e^{\tau _{+}-\tau
_{-}}\prod_{b=1}^{\mathsf{N}}\delta _{b}\gamma _{b}+e^{\tau _{-}-\tau
_{+}}\prod_{b=1}^{\mathsf{N}}\alpha _{b}\beta _{b})}{\left( \zeta
_{+}-1/\zeta _{+}\right) \left( \zeta _{-}-1/\zeta _{-}\right) }.
\end{equation}
\end{lemma}

\begin{proof}
In the previous section, we have shown that the transformations $\lambda
\rightarrow -\lambda $ and $\lambda \rightarrow 1/\lambda $ are symmetries
of the transfer matrix $\mathcal{T}(\lambda )$\ so if $\tau (\lambda )\in
\Sigma _{\mathcal{T}}$ \ then $\tau (\lambda )$ is left unchanged under
these transformations. Moreover, the asymptotic of the transfer matrix can
be easily derived by direct computations, it is central and it holds:%
\begin{equation}
\tau _{\infty }=\lim_{\log \lambda \rightarrow \pm \infty }\lambda ^{\mp 2(%
\mathsf{N}+2)}\mathcal{T}(\lambda ).  \label{asymp-T}
\end{equation}%
The identities (\ref{OpenCyU-identities}) imply after some simple
computation that the transfer matrix is central in $q^{\pm 1/2}$ and $%
iq^{\pm 1/2}$ and that it holds:%
\begin{equation}
\mathcal{T}(q^{\pm 1/2})=\left( -1\right) ^{\mathsf{N}}X\text{det}_{q}M(1),%
\text{ \ }\mathcal{T}(iq^{\pm 1/2})=\left( -1\right) ^{\mathsf{N}}X\frac{%
(\zeta _{+}+1/\zeta _{+})}{(\zeta _{+}-1/\zeta _{+})}\frac{(\zeta
_{-}+1/\zeta _{-})}{(\zeta _{-}-1/\zeta _{-})}\text{det}_{q}M(i).
\end{equation}%
The known functional form of $\mathcal{T}(\lambda )$ w.r.t. $\lambda $
together with this identities imply the interpolation formula in the lemma.
\end{proof}

The previous lemma defines the set of polynomials to which belong the
transfer matrix eigenvalues; in order to completely characterize the
eigenvalues we introduce now the following one-parameter family $D_{\tau
}(\lambda )$ of $p\times p$ matrices:%
\begin{equation}
D_{\tau }(\lambda )\equiv 
\begin{pmatrix}
\tau (\lambda ) & -\text{\textsc{a}}{}(1/\lambda ) & 0 & \cdots & 0 & -\text{%
\textsc{a}}(\lambda ) \\ 
-\text{\textsc{a}}(q\lambda ) & \tau (q\lambda ) & -\text{\textsc{a}}%
(1/\left( q\lambda \right) ) & 0 & \cdots & 0 \\ 
0 & {\quad }\ddots &  &  &  & \vdots \\ 
\vdots &  & \cdots &  &  & \vdots \\ 
\vdots &  &  & \cdots &  & \vdots \\ 
\vdots &  &  &  & \ddots {\qquad } & 0 \\ 
0 & \ldots & 0 & -\text{\textsc{a}}(q^{2l-1}\lambda ) & \tau
(q^{2l-1}\lambda ) & -\text{\textsc{a}}(1/\left( q^{2l-1}\lambda \right) )
\\ 
-\text{\textsc{a}}(1/\left( q^{2l}\lambda \right) ) & 0 & \ldots & 0 & -%
\text{\textsc{a}}(q^{2l}\lambda ) & \tau (q^{2l}\lambda )%
\end{pmatrix}%
,  \label{FrbtD-matrix}
\end{equation}%
where for now $\tau (\lambda )$ is a generic function and we have defined:%
\begin{equation}
\text{\textsc{a}}(\lambda )=\mathsf{a}_{+}(\lambda )\mathsf{A}_{-}(\lambda ).
\end{equation}%
Note that the coefficient \textsc{a}$(\lambda )$ satisfies the quantum
determinant condition:%
\begin{equation}
\text{\textsc{a}}(\lambda q^{1/2})\text{\textsc{a}}(q^{1/2}/\lambda )=\frac{%
\mathsf{a}_{+}(\lambda q^{1/2})\mathsf{a}_{+}(q^{1/2}/\lambda )\text{det}_{q}%
\mathcal{U}_{-}(\lambda )}{\left( \lambda /q\right) ^{2}-\left( q/\lambda
\right) ^{2}}.
\end{equation}

\begin{lemma}
\label{Form-detD}Let $\tau (\lambda )$ be a function of $\lambda $\
invariant under the transformation $\lambda \rightarrow 1/\lambda $ then det$%
_{p}D_{\tau }(\lambda )$ is a function of $\lambda ^{p}$ invariant under the
transformation $\lambda ^{p}\rightarrow 1/\lambda ^{p}$.
\end{lemma}

\begin{proof}
Let us observe that for the invariance of the function $\tau (\lambda )$
under $\lambda \rightarrow 1/\lambda $, we have that:%
\begin{equation}
D_{\tau }(1/\lambda )=O_{C}O_{R}(D_{\tau }(\lambda )),
\end{equation}%
where we have denoted by $O_{R}$ the operation on a $p\times p$ matrix which
exchanges the couple of rows $p-i$ with $i+2$ for any $i\in \{0,...,\left(
p-3\right) /2\}$, similarly $O_{C}$ is the operation on a $p\times p$ matrix
which exchanges the couple of columns $p-i$ with $i+2$ for any $i\in
\{0,...,\left( p-3\right) /2\}$. It is then trivial to see that:%
\begin{equation}
\text{det}_{p}D_{\tau }(1/\lambda )=\text{det}_{p}D_{\tau }(\lambda ).
\end{equation}%
Let us now observe that:%
\begin{equation}
D_{\tau }(\lambda q)=C_{p\rightarrow 1}R_{p\rightarrow 1}(D_{\tau }(\lambda
)),
\end{equation}%
where $R_{p\rightarrow 1}$\ is the operation on a $p\times p$ matrix which
move the last row in the first row leaving the order of the others unchanged
and similarly $C_{p\rightarrow 1}$\ is the operation on a $p\times p$ matrix
which move the last column in the first column leaving the order of the
others unchanged. This clearly implies that:%
\begin{equation}
\text{det}_{p}D_{\tau }(q\lambda )=\text{det}_{p}D_{\tau }(\lambda ),
\end{equation}%
which completes the proof of the lemma.
\end{proof}

As mentioned before, we will restrict our attention to the special boundary
condition $b_{+}(\lambda )=0$, for which the SoV-basis coincides with the $%
\mathcal{B}_{-}$-eigenstates basis. That is we consider the spectral problem
of the transfer matrix: 
\begin{equation}
\mathcal{T}(\lambda )\equiv \mathcal{T}_{\setminus }(\lambda )+c_{+}\left(
\lambda \right) \mathcal{B}_{-}(\lambda ),  \label{OpenCyT-}
\end{equation}%
under the following conditions on the boundary parameters:%
\begin{equation}
b_{+}\left( \lambda \right) =0\text{ and }b_{-}\left( \lambda \right) \neq 0.
\label{OpenCyboundary-}
\end{equation}%
Note that the condition $b_{+}\left( \lambda \right) =0$, keeping instead if
desired a $c_{+}\left( \lambda \right) \neq 0$, can be simply realized by
the following renormalization of the boundary parameters $\kappa
_{+}=e^{-\gamma }\kappa _{+}^{\prime }$\ and $e^{\tau _{+}}=e^{\tau
_{+}^{\prime }-\gamma }$ by sending $\gamma \rightarrow +\infty $. Under
this limit the asymptotic of the transfer matrix reads:%
\begin{equation}
\tau _{\infty }=\frac{(-1)^{\mathsf{N}}\kappa _{+}^{\prime }\kappa
_{-}e^{\tau _{-}-\tau _{+}^{\prime }}\prod_{b=1}^{\mathsf{N}}\alpha
_{b}\beta _{b}}{\left( \zeta _{+}-1/\zeta _{+}\right) \left( \zeta
_{-}-1/\zeta _{-}\right) }.  \label{tau_inf-Tri-K+}
\end{equation}%
In the following we will suppress the unnecessary prime in $\kappa _{+}$ and 
$\tau _{+}$.

\begin{theorem}
\label{OpenCyC:T-eigenstates-}If the conditions:%
\begin{equation}
\mu _{n,+}^{2}\neq q^{-2h}\zeta _{+}^{\pm 2},\text{ }\forall h\in
\{1,...,p-1\},\text{ }n\in \{1,...,\mathsf{N}\}  \label{cond-simple}
\end{equation}
and $\left( \ref{E-SOV}\right) $-$\left( \ref{condition-SoV-2}\right) $ are
satisfied, then $\mathcal{T}(\lambda )$ has simple spectrum and $\Sigma _{%
\mathcal{T}}$ coincides with the set of polynomials $\tau (\lambda )$ of the
form $(\ref{set-tau})$ with $(\ref{tau_inf-Tri-K+})$ which satisfy the
following discrete system of equations:%
\begin{equation}
\text{det}\text{$D$}_{\tau }(\zeta _{a}^{(0)})=0,\text{ }\forall a\in
\{1,...,\mathsf{N}\}.  \label{OpenCyI-Functional-eq}
\end{equation}

\begin{itemize}
\item[\textsf{I)}] The right $\mathcal{T}$-eigenstate corresponding to $\tau
(\lambda )\in \Sigma _{\mathcal{T}}$ is defined by the following
decomposition in the right SoV-basis:%
\begin{equation}
|\tau \rangle =\sum_{h_{1},...,h_{\mathsf{N}}=0}^{p-1}\prod_{a=1}^{\mathsf{N}%
}Q_{\tau ,a}^{(h_{a})}\prod_{1\leq b<a\leq \mathsf{N}%
}(X_{a}^{(h_{a})}-X_{b}^{(h_{b})})|h_{1},...,h_{\mathsf{N}}\rangle ,
\label{OpenCyeigenT-r-D}
\end{equation}%
where the $Q_{\tau ,a}^{(h_{a})}$ are the unique nontrivial solution up to
normalization of the linear homogeneous system:%
\begin{equation}
\text{$D$}_{\tau }(\zeta _{a}^{(0)})\left( 
\begin{array}{c}
Q_{\tau ,a}^{(0)} \\ 
\vdots \\ 
Q_{\tau ,a}^{(p-1)}%
\end{array}%
\right) =\left( 
\begin{array}{c}
0 \\ 
\vdots \\ 
0%
\end{array}%
\right) .  \label{OpenCyt-Q-relation}
\end{equation}

\item[\textsf{II)}] The left $\mathcal{T}$-eigenstate corresponding to $\tau
(\lambda )\in \Sigma _{\mathcal{T}}$ is defined by the following
decomposition in the left SoV-basis:%
\begin{equation}
\langle \tau |=\sum_{h_{1},...,h_{\mathsf{N}}=0}^{p-1}\prod_{a=1}^{\mathsf{N}%
}\hat{Q}_{\tau ,a}^{(h_{a})}\prod_{1\leq b<a\leq \mathsf{N}%
}(X_{a}^{(h_{a})}-X_{b}^{(h_{b})})\langle h_{1},...,h_{\mathsf{N}}|,
\label{OpenCyeigenT-l-D}
\end{equation}%
where the $\hat{Q}_{\tau ,a}^{(h_{a})}$ are the unique nontrivial solution
up to normalization of the linear homogeneous system:%
\begin{equation}
\left( 
\begin{array}{ccc}
\hat{Q}_{\tau ,a}^{(0)} & \ldots & \hat{Q}_{\tau ,a}^{(p-1)}%
\end{array}%
\right) \left( \text{$\hat{D}$}_{\tau }(\zeta _{a}^{(0)})\right)
^{t_{0}}=\left( 
\begin{array}{ccc}
0 & \ldots & 0%
\end{array}%
\right) ,
\end{equation}%
and $\hat{D}$$_{\tau }(\lambda )$ is the family of $p\times p$ matrices
defined substituting in $D_{\tau }(\lambda )$ the coefficient \textsc{a}$%
(\lambda )$ with 
\begin{equation}
\text{\textsc{d}}(\lambda )=\mathsf{d}_{+}(\lambda )\mathsf{D}_{-}(\lambda ).
\end{equation}
\end{itemize}

Finally, using ideas from \cite{OpenCyGMN12-SG}, let us note that if $\tau
(\lambda )\neq \tau ^{\prime }(\lambda )\in \Sigma _{\mathcal{T}}$:%
\begin{equation}
\sum_{b=1}^{\mathsf{N}}\mathcal{M}_{a,b}^{\left( \tau ,\tau ^{\prime
}\right) }x_{b}^{\left( \tau ,\tau ^{\prime }\right) }=0\text{ \ \ \ \ }%
\forall a\in \{1,...,\mathsf{N}\},  \label{Deg-sp-cond}
\end{equation}%
where the $x_{b}^{\left( \tau ,\tau ^{\prime }\right) }$ are defined by:%
\begin{equation}
\tau (\lambda )-\tau ^{\prime }(\lambda )\equiv \left( \Lambda
^{2}-X^{2}\right) \sum_{b=1}^{\mathsf{N}}x_{b}^{\left( \tau ,\tau ^{\prime
}\right) }\Lambda ^{b-1},
\end{equation}%
which in particular implies that the action of $\langle \tau |$ on $|\tau
^{\prime }\rangle $ is zero.
\end{theorem}

\begin{proof}
The spectrum (eigenvalues and eigenstates) of the transfer matrix $\mathcal{T%
}(\lambda )$ in the SoV-basis is characterized by the following discrete
system of equations:%
\begin{equation}
\tau (\xi _{n}^{(h_{n})})\Psi _{\tau }(\text{\textbf{\textit{h}}})\,=\text{\textsc{a}}%
(\xi _{n}^{(h_{n})})\Psi _{\tau }(\mathsf{T}_{n}^{-}(\text{\textbf{\textit{h}}}))+%
\text{\textsc{a}}(1/\xi _{n}^{(h_{n})})\Psi _{\tau }(\mathsf{T}_{n}^{+}(%
\text{\textbf{\textit{h}}})),  \label{SOVBax1}
\end{equation}%
for any$\,n\in \{1,...,\mathsf{N}\}$ and \textbf{\textit{h}}$\in \{0,...,p-1\}^{%
\mathsf{N}}$, i.e. a system of $p^{\mathsf{N}}$ Baxter-like equations in the 
\textit{wave-functions}:%
\begin{equation}
\Psi _{\tau }(\text{\textbf{\textit{h}}})\equiv \langle h_{1},...,h_{\mathsf{N}}|\tau
\rangle ,
\end{equation}%
of the $\mathcal{T}$-eigenstate $|\tau \rangle $ associated to $\tau
(\lambda )\in \Sigma _{\mathcal{T}}$, where:%
\begin{equation}
\mathsf{T}_{n}^{\pm }(\text{\textbf{\textit{h}}})\equiv (h_{1},\dots ,h_{n}\pm
1,\dots ,h_{\mathsf{N}}).
\end{equation}%
This system admits the following equivalent representation as $\mathsf{N}$
linear systems of homogeneous equations:%
\begin{equation}
\text{$D$}_{\tau }(\xi _{n}^{(0)})\left( 
\begin{array}{c}
\Psi _{\tau }(h_{1},...,h_{n}=0,...,h_{\mathsf{N}}) \\ 
\Psi _{\tau }(h_{1},...,h_{n}=1,...,h_{\mathsf{N}}) \\ 
\vdots \\ 
\Psi _{\tau }(h_{1},...,h_{n}=p-1,...,h_{\mathsf{N}})%
\end{array}%
\right) =\left( 
\begin{array}{c}
0 \\ 
0 \\ 
\vdots \\ 
0%
\end{array}%
\right) ,  \label{homo-system}
\end{equation}%
for any$\,n\in \{1,...,\mathsf{N}\}$ and for any $h_{m\neq n}$ in $%
\{0,...,p-1\}$. Then the condition $\tau (\lambda )\in \Sigma _{\mathcal{T}}$
implies the compatibility equations for these linear systems, i.e. it must
hold:%
\begin{equation}
\text{det}\text{$D$}_{\tau }(\xi _{a}^{(0)})=0,\text{ }\forall a\in \{1,...,%
\mathsf{N}\},
\end{equation}%
note that in fact for the previous lemma this conditions are verified also
in the points $\zeta _{a}^{(0)}=1/\xi _{a-\mathsf{N}}^{(0)}$ for any $a\in \{%
\mathsf{N}+1,...,2\mathsf{N}\}$. The rank of the matrices in $\left( \ref%
{homo-system}\right) $ is $p-1$ being 
\begin{equation}
\text{\textsc{a}}(\xi _{n}^{(h)})\neq 0,\text{ \textsc{a}}(1/\xi
_{n}^{(h-1)})\neq 0\ \ \forall h\in \{1,...,p-1\},\text{ }n\in \{1,...,%
\mathsf{N}\},
\end{equation}%
for the conditions $\left( \ref{E-SOV}\right) $, $\left( \ref%
{condition-SoV-2}\right) $ and $\left( \ref{cond-simple}\right) $. Then (up
to an overall normalization) the solution is unique and independent from the 
$h_{m\neq n}\in \{0,...,p-1\}$ for any$\,n\in \{1,...,\mathsf{N}\}$. So
fixing $\tau (\lambda )\in \Sigma _{\mathcal{T}}$ there exists (up to
normalization) one and only one corresponding $\mathcal{T}$-eigenstate $%
|\tau \rangle $ with coefficients of the factorized form given in $\left( %
\ref{OpenCyeigenT-r-D}\right) $-$\left( \ref{OpenCyt-Q-relation}\right) $;
i.e. the $\mathcal{T}$-spectrum is simple.

Vice versa, if $\tau (\lambda )$ is in the set of functions (\ref{set-tau})
and satisfies (\ref{OpenCyI-Functional-eq}), then the state $|\tau \rangle $
defined by $\left( \ref{OpenCyeigenT-r-D}\right) $-$\left( \ref%
{OpenCyt-Q-relation}\right) $ satisfies:%
\begin{equation}
\left\langle h_{1},...,h_{\mathsf{N}}\right\vert \mathcal{T}(\zeta
_{n}^{(h_{n})})|\tau \rangle =\tau (\zeta _{n}^{(h_{n})})\langle
h_{1},...,h_{\mathsf{N}}|\tau \rangle \text{ \ }\forall n\in \{1,...,\mathsf{%
N}\}
\end{equation}%
for any $\mathcal{B}_{-}$-eigenstate $\left\langle h_{1},...,h_{\mathsf{N}%
}\right\vert $ and this implies:%
\begin{equation}
\left\langle h_{1},...,h_{\mathsf{N}}\right\vert \mathcal{T}(\lambda )|\tau
\rangle =\tau (\lambda )\langle h_{1},...,h_{\mathsf{N}}|\tau \rangle \text{
\ \ }\forall \lambda \in \mathbb{C},
\end{equation}%
i.e. $\tau (\lambda )\in \Sigma _{\mathcal{T}}$\ and $|\tau \rangle $ is the
corresponding $\mathcal{T}$-eigenstate.

For the left $\mathcal{T}$-eigenstates the proof follows as above, we
just remark in this case that the matrix elements:%
\begin{equation}
\langle \tau |\mathcal{T}(\zeta _{n}^{(h_{n})})|h_{1},...,h_{\mathsf{N}%
}\rangle ,
\end{equation}%
are computed in the right $\mathcal{B}_{-}$-representation:%
\begin{equation}
\tau (\zeta _{n}^{(h_{n})})\hat{\Psi}_{\tau }(\text{\textbf{\textit{h}}})\,=\text{%
\textsc{d}}(\zeta _{n}^{(h_{n})})\hat{\Psi}_{\tau }(\mathsf{T}_{n}^{-}(\text{%
\textbf{\textit{h}}}))+\text{\textsc{d}}(1/\zeta _{n}^{(h_{n})})\hat{\Psi}_{\tau }(%
\mathsf{T}_{n}^{+}(\text{\textbf{\textit{h}}})),\text{ \ \ }\forall n\in \{1,...,%
\mathsf{N}\}
\end{equation}%
where:%
\begin{equation}
\hat{\Psi}_{\tau }(\text{\textbf{\textit{h}}})\equiv \langle \tau |h_{1},...,h_{%
\mathsf{N}}\rangle ,\text{ \ \textsc{d}}(\lambda )\equiv \mathsf{d}%
_{+}(\lambda )\mathsf{D}_{-}(\lambda ).
\end{equation}%
It is simple to observe that it holds:%
\begin{equation}
\text{det}\text{$D$}_{\tau }(\lambda )=\text{det}\text{$\hat{D}$}_{\tau
}(\lambda ),
\end{equation}%
as a consequence of the following identities:%
\begin{equation}
\text{\textsc{d}}(\lambda )=\alpha (\lambda )\text{\textsc{a}}(q/\lambda )
\end{equation}%
and:%
\begin{equation}
\alpha (1/\lambda )\alpha (q\lambda )=1,\text{ \ }\prod_{a=0}^{p-1}\alpha
(\lambda q^{a})=1
\end{equation}%
where:%
\begin{equation}
\alpha (\lambda )=\frac{s(\lambda )}{s(q/\lambda )}k(\lambda ),\text{ \ }%
s(\lambda )=\frac{\lambda ^{2}q-1/\left( q\lambda ^{2}\right) }{\lambda
^{2}-1/\lambda ^{2}}.
\end{equation}

Finally, the identity (\ref{Deg-sp-cond}) is quite general in the SoV
framework and can be proven also in our case of cyclic representations of
the 6-vertex reflection algebra by following the same proof first given in
the case of a periodic lattice \cite{OpenCyGMN12-SG}.
\end{proof}

For next applications it is interesting to show that we can obtain the
coefficients of a left transfer matrix eigenstates in terms of those of the
right one by introducing a recursion formula that produces both coefficients
in terms of the transfer matrix eigenvalues. The following lemma holds:

\begin{lemma}
Let $\tau (\lambda )\in \Sigma _{\mathcal{T}}$ then we have:%
\begin{equation}
\frac{\hat{Q}_{\tau ,a}^{(h)}}{\hat{Q}_{\tau ,a}^{(h-1)}}=\frac{\text{%
\textsc{a}}(1/\zeta _{a}^{(h-1)})}{\text{\textsc{d}}(1/\zeta _{a}^{(h-1)})}%
\frac{Q_{\tau ,a}^{(h)}}{Q_{\tau ,a}^{(h-1)}},
\end{equation}%
being:%
\begin{equation}
\frac{\hat{Q}_{\tau ,a}^{(h)}}{\hat{Q}_{\tau ,a}^{(0)}}=\frac{t_{\tau
,a}^{(h)}}{\prod_{b=0}^{h-1}\text{\textsc{d}}(1/\zeta _{a}^{(b)})},\text{ \
\ }\frac{Q_{\tau ,a}^{(h)}}{Q_{\tau ,a}^{(0)}}=\frac{t_{\tau ,a}^{(h)}}{%
\prod_{b=0}^{h-1}\text{\textsc{a}}(1/\zeta _{a}^{(b)})},
\label{Recurence-tau-L/R}
\end{equation}%
where the $t_{\tau ,a}^{(h)}$ are defined by the following recursion formula:%
\begin{equation}
t_{\tau ,a}^{(h)}=\tau (\zeta _{a}^{(h-1)})t_{\tau ,a}^{(h-1)}-\frac{\text{%
det}_{q}K_{+}(\xi _{a}^{(h-3/2)})\text{det}_{q}\mathcal{U}_{-}(\xi
_{a}^{(h-3/2)})}{(\xi _{a}^{(h-1/2)})^{2}-1/(\xi _{a}^{(h-1/2)})^{2}}t_{\tau
,a}^{(h-2)}\text{ for }h\in \left\{ 1,...,p-1\right\}
\label{coeff-recurrence}
\end{equation}%
with the following initial conditions $t_{\tau ,a}^{(-1)}=0$, $t_{\tau
,a}^{(0)}=1$.
\end{lemma}

\begin{proof}
We just need to prove the last two identities in this lemma as the first
ones are simple consequences of them. Let us prove the second identity in $%
\left( \ref{Recurence-tau-L/R}\right) $. For $h=1$ we have:%
\begin{equation}
\frac{Q_{\tau ,a}^{(1)}}{Q_{\tau ,a}^{(0)}}=\frac{\tau (\zeta _{a}^{(0)})}{%
\text{\textsc{a}}(1/\zeta _{a}^{(0)})},
\end{equation}%
by the SoV Baxter's like equation and the condition:%
\begin{equation}
\text{\textsc{a}}(\zeta _{a}^{(0)})=0,\text{ \ \textsc{a}}(1/\zeta
_{a}^{(0)})\neq 0.
\end{equation}%
This means that the second formula in $\left( \ref{Recurence-tau-L/R}\right) 
$ and $\left( \ref{coeff-recurrence}\right) $ are both satisfied for $h=1$
once the initial conditions $t_{\tau ,a}^{(-1)}=0$, $t_{\tau ,a}^{(0)}=1$
are imposed. So let us assume that these two identities are satisfied for
any $h\in \left\{ 1,...,x\right\} $ and let us prove it for $h=x+1\leq p-1$.
We have that by the SoV equations it holds:%
\begin{equation}
\frac{Q_{\tau ,a}^{(x+1)}}{Q_{\tau ,a}^{(x-1)}}=\frac{\tau (\zeta _{a}^{(x)})%
}{\text{\textsc{a}}(1/\zeta _{a}^{(x)})}\frac{Q_{\tau ,a}^{(x)}}{Q_{\tau
,a}^{(x-1)}}-\frac{\text{\textsc{a}}(\zeta _{a}^{(x)})}{\text{\textsc{a}}%
(1/\zeta _{a}^{(x)})}
\end{equation}%
and so:%
\begin{equation}
\frac{Q_{\tau ,a}^{(x+1)}}{Q_{\tau ,a}^{(0)}}=\frac{\tau (\zeta _{a}^{(x)})}{%
\text{\textsc{a}}(1/\zeta _{a}^{(x)})}\frac{Q_{\tau ,a}^{(x)}}{Q_{\tau
,a}^{(0)}}-\frac{\text{\textsc{a}}(\zeta _{a}^{(x)})}{\text{\textsc{a}}%
(1/\zeta _{a}^{(x)})}\frac{Q_{\tau ,a}^{(x-1)}}{Q_{\tau ,a}^{(0)}}
\end{equation}%
which by using the formula $\left( \ref{Recurence-tau-L/R}\right) $ for $%
h=x-1$ and $h=x$ reads:%
\begin{equation}
\frac{Q_{\tau ,a}^{(x+1)}}{Q_{\tau ,a}^{(0)}}=\frac{\tau (\zeta
_{a}^{(x)})t_{\tau ,a}^{(x)}-\text{\textsc{a}}(\zeta _{a}^{(x)})\text{%
\textsc{a}}(1/\zeta _{a}^{(x-1)})t_{\tau ,a}^{(x-1)}}{\prod_{b=0}^{x}\text{%
\textsc{a}}(1/\zeta _{a}^{(b)})},
\end{equation}%
which just proves the formulae $\left( \ref{Recurence-tau-L/R}\right) $ and $%
\left( \ref{coeff-recurrence}\right) $ for $h=x+1$ once we recall that:%
\begin{equation}
\text{\textsc{a}}(\zeta _{a}^{(x)})\text{\textsc{a}}(1/\zeta _{a}^{(x-1)})=%
\frac{\text{det}_{q}K_{a,+}(\xi _{a}^{(x-1/2)})\text{det}_{q}\mathcal{U}%
_{a,-}(\xi _{a}^{(x-1/2)})}{(\xi _{a}^{(h+1/2)})^{2}-1/(\xi
_{a}^{(h+1/2)})^{2}}.
\end{equation}%
The proof of the first identity in $\left( \ref{Recurence-tau-L/R}\right) $
can be done in the same way, we have just to use now that:%
\begin{equation}
\text{\textsc{d}}(\zeta _{a}^{(x)})\text{\textsc{d}}(1/\zeta _{a}^{(x-1)})=%
\frac{\text{det}_{q}K_{a,+}(\xi _{a}^{(x-1/2)})\text{det}_{q}\mathcal{U}%
_{a,-}(\xi _{a}^{(x-1/2)})}{(\xi _{a}^{(h+1/2)})^{2}-1/(\xi
_{a}^{(h+1/2)})^{2}}.
\end{equation}
\end{proof}

\section{\label{functional equation}Functional equation characterizing the $%
\mathcal{T}$-spectrum}

In this section we assume that at any quantum site the following constraints
are satisfied:%
\begin{eqnarray}
b_{n}^{p}+a_{n}^{p} &=&0,\text{ \ \ }\forall n\in \{1,...,\mathsf{N}\}, \\
c_{n}^{p}+d_{n}^{p} &=&0,\text{ \ \ }\forall n\in \{1,...,\mathsf{N}\}.
\end{eqnarray}%
Under these conditions we can explicitly construct the left and right basis
which diagonalize the one-parameter family of commuting operators $\mathcal{B%
}_{-}(\lambda )$ and $\mathcal{C}_{-}(\lambda )$ associated to the most
general $K_{-}(\lambda )$ matrix. In the previous section we have done this
construction for the $\mathcal{B}_{-}(\lambda )$ eigenstates basis, clearly
we can do a similar construction also for $\mathcal{C}_{-}(\lambda )$. In
the previous sections we have explained how the spectral problem of the
transfer matrix $\mathcal{T}(\lambda )$ can be characterized by SoV in the $%
\mathcal{B}_{-}(\lambda )$ eigenstates basis when $b_{+}\left( \lambda
\right) =0$ keeping instead if desired a $c_{+}\left( \lambda \right) \neq 0$%
, by the same approach we can characterize the $\mathcal{T}(\lambda )$%
-spectrum by SoV in the $\mathcal{C}_{-}(\lambda )$ eigenstates basis when $%
c_{+}\left( \lambda \right) =0$ keeping instead if desired a $b_{+}\left(
\lambda \right) \neq 0$.

Here, we show that the SoV characterization of the $\mathcal{T}(\lambda )$%
-spectrum can be reformulated by functional equations. Before this, let us
just observe that the central value of the asymptotic of $\mathcal{T}%
(\lambda )$ is given by:%
\begin{equation}
\tau _{\infty }=\frac{(-1)^{\mathsf{N}}\kappa _{+} \kappa
_{-}e^{\epsilon (\tau _{-}-\tau _{+} )}\prod_{b=1}^{\mathsf{N}%
}\alpha _{b}\beta _{b}}{\left( \zeta _{+}-1/\zeta _{+}\right) \left( \zeta
_{-}-1/\zeta _{-}\right) }\text{,}
\end{equation}%
where $\epsilon =+1$ for $b_{+}\left( \lambda \right) =0$ (keeping instead a 
$c_{+}\left( \lambda \right) \neq 0$) and $\epsilon =-1$ for $c_{+}\left(
\lambda \right) =0$ (keeping instead a $b_{+}\left( \lambda \right) \neq 0$%
). Let us moreover introduce the following notations:%
\begin{eqnarray}
\overline{\text{\textsc{a}}}(\lambda ) &=&\lambda q^{-1/2}\text{\textsc{a}}%
(\lambda ), \\
\overline{\text{\textsc{a}}}_{\infty } &=&\lim_{\lambda \rightarrow +\infty
}\lambda ^{-2(\mathsf{N}+2)}\overline{\text{\textsc{a}}}(\lambda )=\frac{%
(-1)^{\mathsf{N}}\alpha _{-}\beta _{-}\zeta _{+}\kappa _{-}\prod_{n=1}^{%
\mathsf{N}}b_{n}c_{n}}{q^{1+\mathsf{N}}\left( \zeta _{+}-1/\zeta _{+}\right)
\left( \zeta _{-}-1/\zeta _{-}\right) },
\end{eqnarray}%
and:%
\begin{equation}
F(\lambda )=\prod_{b=1}^{2\mathsf{N}}\left( \frac{\lambda ^{p}}{\left( \zeta
_{b}^{(0)}\right) ^{p}}-\frac{\left( \zeta _{b}^{(0)}\right) ^{p}}{\lambda
^{p}}\right) ,
\end{equation}%
then the following results hold:

\begin{proposition}
If the conditions $\left( \ref{E-SOV}\right) $, $\left( \ref{condition-SoV-2}%
\right) $ and $\left( \ref{cond-simple}\right) $ are satisfied and:%
\begin{equation}
b_{n}=-q^{2j_{n}-1}a_{n},\text{ \ }d_{n}=-q^{2j_{n}-1}c_{n},
\label{Double-nilp}
\end{equation}%
then $\mathcal{T}(\lambda )$ has simple spectrum and $\tau (\lambda )$ of
the form $(\ref{set-tau})$ with $(\ref{tau_inf-Tri-K+})$ is an element of $%
\Sigma _{\mathcal{T}}$ \ if and only if det$_{p}\bar{D}_{\tau }(\lambda )$
is a Laurent polynomial of degree $\mathsf{N}+2$ in the variable:%
\begin{equation}
Z=\lambda ^{2p}+\frac{1}{\lambda ^{2p}}
\end{equation}%
which satisfies the following functional equation:%
\begin{equation}
\text{det}_{p}\bar{D}_{\tau }(\lambda )-F(\lambda )\sum_{a=0}^{1}\text{det}%
_{p}\bar{D}_{\tau }(i^{a}q^{1/2})\frac{\left( \lambda ^{p}+\left( -1\right)
^{a}/\lambda ^{p}\right) ^{2}}{4\left( -1\right) ^{a}F(i^{a}q^{1/2})}=\left(
\tau _{\infty }^{p}-\overline{\text{\textsc{a}}}_{\infty }^{p}\right)
F(\lambda )\left( \lambda ^{2p}-\frac{1}{\lambda ^{2p}}\right) ^{2}.
\label{Func-EQ-1}
\end{equation}%
Here $\bar{D}_{\tau }(\lambda )$ is obtained from $D_{\tau }(\lambda )$ by
substituting in it \textsc{a}$(\lambda )$ with $\overline{\text{\textsc{a}}}%
(\lambda )$.
\end{proposition}

\begin{proof}
Let us observe that det$_{p}\bar{D}_{\tau }(\lambda )$ is an even function
of $\lambda $ as a consequence of the parity of $\tau (\lambda )$ and $%
\overline{\text{\textsc{a}}}(\lambda )$ moreover following the same steps of
Lemma \ref{Form-detD} we can prove that det$_{p}\bar{D}_{\tau }(\lambda )$
is invariant under the transformations $\lambda \rightarrow 1/\lambda $ and $%
\lambda \rightarrow q\lambda $, so that det$_{p}\bar{D}_{\tau }(\lambda )$
is indeed a function of $Z=\lambda ^{2p}+1/\lambda ^{2p}$. Let us now
observe that in the points $\lambda =q^{1/2}$ and $\lambda =iq^{1/2}$, the
central row of the matrix $\bar{D}_{\tau }(\lambda )$ has two elements which
are divergent as proportional to $\overline{\text{\textsc{a}}}(\pm q^{p/2})$
and $\overline{\text{\textsc{a}}}(\pm iq^{p/2})$, respectively. In the
following we prove that:%
\begin{equation}
\text{det}_{p}\text{$\bar{D}$}_{\tau }(q^{1/2})\text{ and det}_{p}\text{$%
\bar{D}$}_{\tau }(iq^{1/2})\text{ are finite,}  \label{finite D_tau}
\end{equation}%
if we ask that the function $\tau (\lambda )$ has the functional form $(\ref%
{set-tau})$. Let us first introduce the notations:%
\begin{equation}
\overline{\text{\textsc{x}}}(\lambda )=\left( \lambda ^{2}-\frac{1}{\lambda
^{2}}\right) \overline{\text{\textsc{a}}}(\lambda ),
\end{equation}%
and using it let us expand the determinant: 
\begin{align}
\text{det}_{p}\text{$\bar{D}$}_{\tau }(\lambda q^{1/2})& =\tau (\lambda )%
\text{det}_{p-1}\text{$\bar{D}$}_{\tau ,(p+1)/2,(p+1)/2}(\lambda q^{1/2})+%
\frac{\overline{\text{\textsc{x}}}(\lambda )\text{det}_{p-1}\text{$\bar{D}$}%
_{\tau ,(p+1)/2,(p+1)/2-1}(\lambda q^{1/2})}{\lambda ^{2}-1/\lambda ^{2}} 
\notag \\
& -\frac{\overline{\text{\textsc{x}}}(1/\lambda )\text{det}_{p-1}\text{$\bar{%
D}$}_{\tau ,(p+1)/2,(p+1)/2+1}(\lambda q^{1/2})}{\lambda ^{2}-1/\lambda ^{2}}%
,
\end{align}%
w.r.t. the central row. Here, we have denoted with $\bar{D}_{\tau
,i,j}(\lambda )$ the $(p-1)\times (p-1)$ matrix defined by removing the row $%
i$ and the column $j\ $to the matrix $\bar{D}$$_{\tau }(\lambda )$. The
following identity holds:%
\begin{equation}
\text{det}_{p-1}\text{$\bar{D}$}_{\tau ,(p+1)/2,(p+1)/2+1}(\lambda q^{1/2})=%
\text{det}_{p-1}\text{$\bar{D}$}_{\tau ,(p+1)/2,(p+1)/2-1}(q^{1/2}/\lambda ),
\end{equation}%
and it follows just exchanging the row $j$ with the row $p-j$, for any $j\in
\left\{ 1,..,\left( p+1\right) /2\right\} $, and then the column $j$ with
the column $p-j$, for any $j\in \left\{ 1,..,\left( p+1\right) /2\right\} $,
in $\bar{D}$$_{\tau ,(p+1)/2,(p+1)/2+1}(\lambda q^{1/2})$. Note that det$%
_{p-1}\bar{D}$$_{\tau ,(p+1)/2,(p+1)/2+1}(\lambda q^{1/2})$ and det$_{p-1}%
\bar{D}$$_{\tau ,(p+1)/2,(p+1)/2-1}(\lambda q^{1/2})$ are Laurent's rational
functions both finite for $\lambda \rightarrow \pm 1$ and $\lambda
\rightarrow \pm i$. This implies that in both the limits $\lambda
\rightarrow \pm 1$ and $\lambda \rightarrow \pm i$ the function det$_{p}\bar{%
D}$$_{\tau }(\lambda q^{1/2})$ is finite. As $Z=2$ and $Z=-2$ are the only
points for which det$_{p}\bar{D}_{\tau }(\lambda )$ may have divergencies,
then, from $\tau (\lambda )$ and $\overline{\text{\textsc{x}}}(\lambda )$
Laurent's polynomial in $\lambda $ of degree $2\mathsf{N}+4$, it follows
that det$_{p}\bar{D}_{\tau }(\lambda )$ is a polynomial of degree $\mathsf{N}%
+2$ in the variable $Z$. Let us now remark that by the SoV characterization
of the spectrum we have that $\tau (\lambda )\in \Sigma _{\mathcal{T}}$ if
and only if it has the form $(\ref{set-tau})$ and satisfies:%
\begin{equation}
\text{det}_{p}\text{$\bar{D}$}_{\tau }(\xi _{a}^{(0)})=0,\text{ }\forall
a\in \{1,...,\mathsf{N}\},
\end{equation}%
as we are assuming that $\left( \ref{Double-nilp}\right) $ is satisfied.
Finally, it is simple to verify that the following asymptotic hold:%
\begin{align}
\text{det}_{p}\left[ \lim_{\lambda \rightarrow +\infty }\lambda ^{-2(\mathsf{%
N}+2)}\bar{D}_{\tau }(\lambda )\right] & =\text{det}_{p}\left[ \lim_{\lambda
\rightarrow 0}\lambda ^{2(\mathsf{N}+2)}\bar{D}_{\tau }(\lambda )\right] ^{t}
\\
& =\text{det}_{p}%
\begin{pmatrix}
\tau _{\infty } & 0 & 0 & \cdots & 0 & -\overline{\text{\textsc{a}}}_{\infty
} \\ 
-x\overline{\text{\textsc{a}}}_{\infty } & x\tau _{\infty } & 0 & 0 & \cdots
& 0 \\ 
0 & -x^{2}\overline{\text{\textsc{a}}}_{\infty } & x^{2}\tau _{\infty } & 
\ddots &  & \vdots \\ 
\vdots &  & \ddots & \ddots & 0 & 0 \\ 
0 & \ldots & 0 & -x^{2l-1}\overline{\text{\textsc{a}}}_{\infty } & 
x^{2l-1}\tau _{\infty } & 0 \\ 
0 & 0 & \ldots & 0 & -x^{2l}\overline{\text{\textsc{a}}}_{\infty } & 
x^{2l}\tau _{\infty }%
\end{pmatrix}%
,
\end{align}%
where we have denoted with $^{t}$ the transpose of the matrix and $x=q^{2(%
\mathsf{N}+2)}$, so that it holds:%
\begin{equation}
\lim_{\log \lambda \rightarrow \pm \infty }\lambda ^{\mp 2p(\mathsf{N}+2)}%
\text{det}_{p}\bar{D}_{\tau }(\lambda )=\tau _{\infty }^{p}-\overline{\text{%
\textsc{a}}}_{\infty }^{p}.
\end{equation}%
These results fix completely the Laurent's polynomial det$_{p}\bar{D}_{\tau
}(\lambda )$ to satisfy $\left( \ref{Func-EQ-1}\right) $.
\end{proof}

Let us introduce now the following function:%
\begin{align}
G(\lambda |x,y)& =F(\lambda )\left[ \left( \tau _{\infty }-q^{-2(p-1)\mathsf{%
N}}\overline{\text{\textsc{a}}}_{\infty }\right) x\prod_{a=0}^{1}\left( 
\frac{\lambda }{i^{a}q^{1/2}}-\frac{i^{a}q^{1/2}}{\lambda }\right) \left(
i^{a}q^{1/2}\lambda -\frac{1}{i^{a}q^{1/2}\lambda }\right) \right.  \notag \\
& \left. +(i-1)\frac{y\text{\textsc{a}}(iq^{1/2})}{4F(iq^{1/2})}%
\prod_{a=0}^{1}\left( \lambda q^{(1-2a)/2}-\frac{1}{\lambda q^{(1-2a)/2}}%
\right) \right] ,
\end{align}%
and the states:%
\begin{eqnarray}
\left\langle \omega \right\vert &=&\sum_{h_{1},...,h_{\mathsf{N}%
}=0}^{p-1}\prod_{a=1}^{\mathsf{N}}\prod_{k_{a}=0}^{h_{a}-1}\frac{\text{%
\textsc{a}}(1/\zeta _{a}^{(k_{a})})}{\text{\textsc{d}}(1/\zeta
_{a}^{(k_{a})})}\prod_{1\leq b<a\leq \mathsf{N}%
}(X_{a}^{(h_{a})}-X_{b}^{(h_{b})})\left\langle h_{1},...,h_{\mathsf{N}%
}\right\vert , \\
|\bar{\omega}\rangle &=&\sum_{h_{1},...,h_{\mathsf{N}}=0}^{p-1}\prod_{1\leq
b<a\leq \mathsf{N}}(X_{a}^{(h_{a})}-X_{b}^{(h_{b})})|h_{1},...,h_{\mathsf{N}%
}\rangle ,
\end{eqnarray}%
and a renormalization of the operator%
\begin{equation}
\mathcal{\hat{B}}_{-}(\lambda )=\frac{(\zeta _{-}-1/\zeta _{-})}{\kappa
_{-}e^{\tau _{-}}(\lambda ^{2}/q-q/\lambda ^{2})\prod_{n=1}^{\mathsf{N}%
}\alpha _{n}\beta _{n}}\mathcal{B}_{-}(\lambda ),
\end{equation}%
which is a degree $\mathsf{N}$ polynomial in $\Lambda $. In the following we
denote with $Q(\lambda )$ a polynomial in $\Lambda =\lambda ^{2}+1/\lambda
^{2}$ of degree $\mathsf{N}_{Q}$ $\leq (p-1)\mathsf{N}$ which admits the
following interpolation formula:%
\begin{equation}
Q(\lambda )=\sum_{a=1}^{(p-1)\mathsf{N}+1}\prod_{b=1,b\neq a}^{(p-1)\mathsf{N%
}+1}\frac{\Lambda -w_{b}}{w_{a}-w_{b}}Q(\xi _{a}),  \label{Q-form}
\end{equation}%
where $w_{b}=$ $\xi _{b}^{2}+1/\xi _{b}^{2}$ for any $b\in \{1,...,(p-1)%
\mathsf{N}+1\}$ and where we have defined:%
\begin{equation}
\xi _{s(n,h_{n})}\equiv \xi _{n}^{(h_{n}-1)},\text{\ }s(n,h_{n})\equiv
\left( n-1\right) \left( p-1\right) +h_{n},\text{ }\forall n\in \{1,...,%
\mathsf{N}\},h_{n}\in \{1,...,p-1\},
\end{equation}%
while $\xi _{(p-1)\mathsf{N}+1}$ is arbitrary. Moreover, we use the notation:%
\begin{equation}
q_{\infty }\equiv \sum_{a=1}^{(p-1)\mathsf{N}+1}\prod_{b=1,b\neq a}^{(p-1)%
\mathsf{N}+1}\frac{Q(\xi _{a})}{w_{a}-w_{b}},\text{ \ }q_{0}\equiv
Q(iq^{1/2}).
\end{equation}%
Here, $q_{\infty }$ is the coefficient in $\Lambda ^{(p-1)\mathsf{N}}$ of
the power expansion of the polynomial $Q(\lambda )$. Once this notation are
introduced, the previous characterization of the spectrum can be
reformulated in terms of Baxter's type TQ-functional equations and the
eigenstates admit an algebraic Bethe ansatz like reformulation, as we show
in the next theorem.

\begin{theorem}
Let the conditions $\left( \ref{E-SOV}\right) $, $\left( \ref%
{condition-SoV-2}\right) $, $\left( \ref{cond-simple}\right) $ and $\left( %
\ref{Double-nilp}\right) $ be satisfied and let $\tau (\lambda )$ be an
entire function for which there exists a polynomial $Q(\lambda )$ of the
form $\left( \ref{Q-form}\right) $ satisfying the conditions:%
\begin{equation}
Q(\xi _{a}^{\left( 0\right) })\neq 0,\text{ }\forall a\in \{1,...,\mathsf{N}%
\},  \label{Q-condition}
\end{equation}%
and the following functional equation:%
\begin{equation}
\tau (\lambda )Q(\lambda )=\overline{\text{\textsc{a}}}(\lambda )Q(\lambda
/q)+\overline{\text{\textsc{a}}}(1/\lambda )Q(\lambda q)+G(\lambda
|q_{\infty },q_{0}),  \label{Inho-Baxter-EQ}
\end{equation}%
with $(p-1)\mathsf{N}-1\leq \mathsf{N}_{Q}$, then $\tau (\lambda )\in \Sigma
_{\mathcal{T}}$ \ and (up to normalization) the left and right transfer
matrix eigenstates associated to it admit the following Bethe ansatz like
representations:%
\begin{equation}
\left\langle \tau \right\vert =\left\langle \omega \right\vert \prod_{b=1}^{%
\mathsf{N}_{Q}}\mathcal{\hat{B}}_{-}(\lambda _{b}),\text{ \ \ }|\tau \rangle
=\prod_{b=1}^{\mathsf{N}_{Q}}\mathcal{\hat{B}}_{-}(\lambda _{b})|\bar{\omega}%
\rangle ,  \label{Bethe-like-eigenstates}
\end{equation}%
where the $\lambda _{b}$ (fixed up the symmetry $\lambda _{b}\rightarrow
-\lambda _{b},$ $\lambda _{b}\rightarrow 1/\lambda _{b}$) for $b\in \{1,...,%
\mathsf{N}_{Q}\}$ are the zeros of $Q(\lambda )$. Vice versa, if $\tau
(\lambda )\in \Sigma _{\mathcal{T}}$ then there exists a polynomial $%
Q(\lambda )$ of the form $\left( \ref{Q-form}\right) $ with $(p-1)\mathsf{N}%
-1\leq \mathsf{N}_{Q}$ satisfying with $\tau (\lambda )$ the Baxter's
equation $\left( \ref{Inho-Baxter-EQ}\right) $.
\end{theorem}

\begin{proof}
The proof of this type of reformulation of the spectrum is by now quite
standard and it has been proven for several models once they admit SoV
description, see for example \cite%
{OpenCyNT-10,OpenCyN-10,OpenCyGN12,OpenCyKitMN14,OpenCyNicT15-2,OpenCyLevNT15,OpenCyNicT15}%
. So we will try to point out just some main features of the proof. Let us
start proving the first part of the statement. It is simple to remark that
the r.h.s of the equation $\left( \ref{Inho-Baxter-EQ}\right) $ is a Laurent
polynomial in $\lambda $, indeed we can write:%
\begin{equation}
\overline{\text{\textsc{a}}}(\lambda )Q(\lambda /q)+\overline{\text{\textsc{a%
}}}(1/\lambda )Q(\lambda q)=\frac{\overline{\text{\textsc{x}}}(\lambda
)Q(q/\lambda )-\overline{\text{\textsc{x}}}(1/\lambda )Q(\lambda q)}{\lambda
^{2}-1/\lambda ^{2}}
\end{equation}%
so that the limits $\lambda \rightarrow \pm 1,$ $\lambda \rightarrow \pm i$
are all finite. Moreover, it is simple to observe that the r.h.s of $\left( %
\ref{Inho-Baxter-EQ}\right) $ is invariant under the transformations $%
\lambda \rightarrow -\lambda $ and $\lambda \rightarrow 1/\lambda $, so that
the r.h.s. of $\left( \ref{Inho-Baxter-EQ}\right) $ is in fact a polynomial
of degree $\mathsf{N}_{Q}+\mathsf{N}+2\leq p\mathsf{N}+2$ in $\Lambda $.
Then, the fact that $\tau (\lambda )$ is entire in $\lambda $ implies by the
equation $\left( \ref{Inho-Baxter-EQ}\right) $ that $\tau (\lambda )$ is a
polynomial in $\Lambda $ of the form $(\ref{set-tau})$ with $(\ref%
{tau_inf-Tri-K+})$. This together with the equations:%
\begin{equation}
\text{det}_{p}\text{$\bar{D}$}_{\tau }(\xi _{a}^{(0)})=0,\text{ }\forall
a\in \{1,...,\mathsf{N}\},
\end{equation}%
which are trivial consequences of $\left( \ref{Inho-Baxter-EQ}\right) $ and
of $\left( \ref{Q-condition}\right) $, imply by the SoV characterization
that $\tau (\lambda )\in \Sigma _{\mathcal{T}}$. Let us show now that the
eigenstates associated to this $\tau (\lambda )\in \Sigma _{\mathcal{T}}$
can be written in the form $\left( \ref{Bethe-like-eigenstates}\right) $. In
order to do so let us observe that the $Q(\lambda )$ which is solution of $%
\left( \ref{Inho-Baxter-EQ}\right) $ is defined up to a constant factor so
that we are free to fix it by writing:%
\begin{equation}
Q(\lambda )=\prod_{b=1}^{\mathsf{N}_{Q}}(\Lambda -\Lambda _{b})
\end{equation}%
where $\Lambda _{b}=\lambda _{b}^{2}+1/\lambda _{b}^{2}$, then we have that
it holds:%
\begin{eqnarray}
\prod_{b=1}^{\mathsf{N}_{Q}}\mathcal{\hat{B}}_{-}(\lambda _{b})|\bar{\omega}%
\rangle &=&\sum_{h_{1},...,h_{\mathsf{N}}=0}^{p-1}\prod_{b=1}^{\mathsf{N}%
_{Q}}\prod_{a=1}^{\mathsf{N}}(X_{a}^{(h_{a})}-\Lambda _{b})\prod_{1\leq
b<a\leq \mathsf{N}}(X_{a}^{(h_{a})}-X_{b}^{(h_{b})})|h_{1},...,h_{\mathsf{N}%
}\rangle  \notag \\
&=&\sum_{h_{1},...,h_{\mathsf{N}}=0}^{p-1}\prod_{a=1}^{\mathsf{N}%
}Q(X_{a}^{(h_{a})})\prod_{1\leq b<a\leq \mathsf{N}%
}(X_{a}^{(h_{a})}-X_{b}^{(h_{b})})|h_{1},...,h_{\mathsf{N}}\rangle ,
\end{eqnarray}%
\newline
where in the last line appears the SoV characterization of the right
eigenstate associated to $\tau (\lambda )\in \Sigma _{\mathcal{T}}$. The
same steps are used to prove $\left( \ref{Bethe-like-eigenstates}\right) $
for the left eigenstate. Let us comment that the polynomiality of $Q(\lambda
)$ is the central property that we have used to prove these rewriting of the
SoV characterizations of the transfer matrix eigenstates. This type of
rewriting was first observe in the case of the noncompact XXX chains \cite%
{DerKM03}. It is quite general and it has been used already for different
models in the SoV framework, see for example \cite%
{OpenCyNicT15-2,OpenCyLevNT15,OpenCyNicT15}.

Let us now assume that $\tau (\lambda )\in \Sigma _{\mathcal{T}}$\ and let
us prove that it exists $Q(\lambda )$ of the form $\left( \ref{Q-form}%
\right) $. In order to do so it is enough to show that the Baxter's equation
holds in $p\mathsf{N}+2$ different values of $\Lambda $, that we chose to be 
$\Lambda =\pm (q+1/q)$ and the $\Lambda =X_{a}^{(h_{a})}$ for any $a\in
\{1,...,\mathsf{N}\}$ and $h_{a}$ $\in \{0,...,p-1\}$. This set of equations
can be organized in the following form:%
\begin{equation}
\text{$\bar{D}$}_{\tau }(\xi _{a}^{(0)})\left( 
\begin{array}{l}
Q(\xi _{a}^{(h=0)}) \\ 
Q(\xi _{a}^{(h=1)}) \\ 
\vdots \\ 
\vdots \\ 
Q(\xi _{a}^{(h=p-1)})%
\end{array}%
\right) _{p\times 1}=\left( 
\begin{array}{l}
0 \\ 
\vdots \\ 
\vdots \\ 
\vdots \\ 
0%
\end{array}%
\right) _{p\times 1}\text{ \ \ \ }\forall a\in \{1,...,\mathsf{N}\}.
\end{equation}%
They are equivalent to the following system of equations:%
\begin{align}
Q(\xi _{a}^{(0)})\tau (\xi _{a}^{(0)})/\overline{\text{\textsc{a}}}(1/\xi
_{a}^{(0)})& =Q(\xi _{a}^{(1)})  \label{System2-1s} \\
\frac{\tau (\xi _{a}^{(h)})}{\overline{\text{\textsc{a}}}(1/\xi _{a}^{(h)})}%
Q(\xi _{a}^{(h)})-\frac{\overline{\text{\textsc{a}}}(\xi _{a}^{(h)})}{%
\overline{\text{\textsc{a}}}(1/\xi _{a}^{(h)})}Q(\xi _{a}^{(h-1)})& =Q(\xi
_{a}^{(h+1)}),\text{\ }\forall h\in \{1,...,p-3\}, \\
\frac{\tau (\xi _{a}^{(p-2)})}{\overline{\text{\textsc{a}}}(1/\xi
_{a}^{(p-2)})}Q(\xi _{a}^{(p-2)})-\frac{\overline{\text{\textsc{a}}}(\xi
_{a}^{(p-2)})}{\overline{\text{\textsc{a}}}(1/\xi _{a}^{(p-2)})}Q(\xi
_{a}^{(p-3)})& =\prod_{b=1}^{(p-1)\mathsf{N}}\frac{X_{a}^{(p-1)}-w_{b}}{%
w_{(p-1)\mathsf{N}+1}-w_{b}}Q(\xi _{(p-1)\mathsf{N}+1})  \notag \\
& +\sum_{n=1}^{\mathsf{N}_{Q}}\prod_{\substack{ b=1  \\ b\neq n}}^{(p-1)%
\mathsf{N}+1}\frac{X_{a}^{(p-1)}-w_{b}}{w_{n}-w_{b}}Q(\xi _{n}),
\label{System2-3s}
\end{align}%
as $\tau (\lambda )\in \Sigma _{\mathcal{T}}$. So we are left with $\left( %
\ref{System2-1s}\right) $-$\left( \ref{System2-3s}\right) $ a linear system
of $(p-1)\mathsf{N}$ inhomogeneous equations in the $(p-1)\mathsf{N}$
unknowns $Q(\xi _{a})$ for all $a\in \{1,...,(p-1)\mathsf{N}\}$. Let us
define by induction the following coefficients:%
\begin{equation}
\mathbb{C}_{a,h+1}=x_{a,0,h}\mathbb{C}_{a,h}+x_{a,-,h}\mathbb{C}_{a,h-1},%
\text{ \ }\forall h\in \{1,...,p-2\},\text{ }\mathbb{C}_{a,1}=\tau (\xi
_{a}^{(0)})/\overline{\text{\textsc{a}}}(1/\xi _{a}^{(0)}),\text{ \ }\mathbb{%
C}_{a,0}=1,
\end{equation}%
where we have defined:%
\begin{equation}
x_{a,0,h}=\frac{\tau (\xi _{a}^{(h)})}{\overline{\text{\textsc{a}}}(1/\xi
_{a}^{(h)})},\text{ \ \ }x_{a,-,h}=-\frac{\overline{\text{\textsc{a}}}(\xi
_{a}^{(h)})}{\overline{\text{\textsc{a}}}(1/\xi _{a}^{(h)})},\text{ \ }%
\forall h\in \{1,...,p-2\},\text{\ }\forall a\in \{1,...,\mathsf{N}\}.
\end{equation}%
One can rewrite the previous system as it follows:%
\begin{equation}
Q(\xi _{a}^{(h)})=\mathbb{C}_{a,h}Q(\xi _{a}^{(0)}),\text{ \ }\forall h\in
\{1,...,p-1\},\text{\ }\forall a\in \{1,...,\mathsf{N}\},
\end{equation}%
and%
\begin{equation}
\mathbb{C}_{a,p-1}Q(\xi _{a}^{(0)})=\prod_{b=1}^{(p-1)\mathsf{N}}\frac{%
X_{a}^{(p-1)}-w_{b}}{w_{(p-1)\mathsf{N}+1}-w_{b}}Q(\xi _{(p-1)\mathsf{N}%
+1})+\sum_{n=1}^{\mathsf{N}}\mathbb{D}_{a,n}Q(\xi _{n}^{(0)}),
\end{equation}%
where%
\begin{equation}
\mathbb{D}_{b,n}=\sum_{h=0}^{p-2}\prod_{_{c=1,c\neq a(n,h)}}^{(p-1)\mathsf{N}%
+1}\frac{X_{b}^{(p-1)}-w_{c}}{w_{n}-w_{c}}\mathbb{C}_{n,h}.
\end{equation}%
So finally we are left with a system of $\mathsf{N}$ inhomogeneous equations
in the $\mathsf{N}$ unknown $Q(\xi _{a}^{(0)})$ for all $a\in \{1,...,%
\mathsf{N}\}$, which always admits a nontrivial solution, i.e. $\{Q(\xi
_{1}^{(0)}),...,Q(\xi _{\mathsf{N}}^{(0)})\}\neq \{0,...,0\}$. Finally, let
us point out that the degree of the polynomial $Q(\lambda )$ is constrained
by the Baxter's equation $\left( \ref{Inho-Baxter-EQ}\right) $. Indeed, it
is trivial to remark that if $\mathsf{N}_{Q}\leq (p-1)\mathsf{N}-2$, then
the equation $\left( \ref{Inho-Baxter-EQ}\right) $ admits only the trivial
solution $Q(\lambda )=0$ which is not compatible with $\{Q(\xi
_{1}^{(0)}),...,Q(\xi _{\mathsf{N}}^{(0)})\}\neq \{0,...,0\}$. Instead, the
equation $\left( \ref{Inho-Baxter-EQ}\right) $\ may still be satisfied with
a nontrivial $Q(\lambda )$ if $\mathsf{N}_{Q}=(p-1)\mathsf{N}-1$ but only if
the following condition:%
\begin{equation}
\lim_{\Lambda \rightarrow \infty }\Lambda ^{(p-1)\mathsf{N}-1}Q(\lambda )=%
\frac{(i-1)\text{\textsc{a}}(iq^{1/2})q_{0}}{4\left( \tau _{\infty
}-q^{2-2(p-1)\mathsf{N}}\overline{\text{\textsc{a}}}_{\infty }\right)
F(iq^{1/2})},
\end{equation}%
is satisfied. It is clear that this represents one supplementary condition
and one can expect that up to some exceptional case, related to the values
of the parameters of the representation and to some special choice of the $%
\tau (\lambda )\in \Sigma _{\mathcal{T}}$, it is not satisfied and so that $%
\mathsf{N}_{Q}=(p-1)\mathsf{N}$.
\end{proof}

It is then simple to prove the following corollary which provides under some
further constraints a complete characterization of the transfer matrix
eigenvalues and eigenstates in terms of solution of ordinary Bethe equations.

\begin{corollary}
If the conditions $\left( \ref{E-SOV}\right) $, $\left( \ref{condition-SoV-2}%
\right) $, $\left( \ref{cond-simple}\right) $ and $\left( \ref{Double-nilp}%
\right) $ are satisfied and if we fix the boundary parameter $\zeta
_{+}=iz_{+}$ with $z_{+}\in \left\{ -1,+1\right\} $ and the following global
condition:%
\begin{equation}
\frac{\kappa _{+}e^{\epsilon (\tau _{-}-\tau _{+})}}{iz_{+}\alpha _{-}\beta
_{-}}=q^{1+\mathsf{N}}\prod_{n=1}^{\mathsf{N}}\frac{b_{n}c_{n}}{\alpha
_{n}\beta _{n}},  \label{Condition-global-boundary}
\end{equation}%
where $\epsilon =+1$ for $b_{+}\left( \lambda \right) =0$, $c_{+}\left(
\lambda \right) \neq 0$ and $\epsilon =-1$ for $c_{+}\left( \lambda \right)
=0$, $b_{+}\left( \lambda \right) \neq 0$, then $\mathcal{T}(\lambda )$ has
simple spectrum and $\tau (\lambda )\in \Sigma _{\mathcal{T}}$ if and only
if $\tau (\lambda )$ is entire and there exists a polynomial $Q(\lambda )$
of the form $\left( \ref{Q-form}\right) $, with $\mathsf{N}_{Q}\leq (p-1)%
\mathsf{N}$, which satisfies the following homogeneous Baxter equation:%
\begin{equation}
\tau (\lambda )Q(\lambda )=\overline{\text{\textsc{a}}}(\lambda )Q(\lambda
/q)+\overline{\text{\textsc{a}}}(1/\lambda )Q(\lambda q).
\label{Baxter-hom-eq}
\end{equation}%
The $(\lambda _{1},...,\lambda _{\mathsf{N}_{Q}})$ entering in the Bethe
ansatz like representations of the eigenstates $\left( \ref%
{Bethe-like-eigenstates}\right) $ are solutions of the associated ordinary
Bethe equations.
\end{corollary}

\begin{proof}
The condition $\zeta _{+}=iz_{+}$ with $z_{+}\in \left\{ -1,+1\right\} $
implies:%
\begin{equation}
\tau (iq^{1/2})=\text{\textsc{a}}(iq^{1/2})=0,
\end{equation}%
so that the condition $\left( \ref{Condition-global-boundary}\right) $
implies that $G(\lambda |q_{\infty },q_{1})=0$, for any choice of $q_{\infty
},$ $q_{1}$. Then, the previous theorem directly implies our corollary.
Finally, let us remark that in this case we do not have any restriction on
the degree of the polynomial $Q(\lambda )$ imposed by the Baxter's equation $%
\left( \ref{Inho-Baxter-EQ}\right) $, so that we are only left with $\mathsf{%
N}_{Q}\leq (p-1)\mathsf{N}$.
\end{proof}

\section{Conclusions}

In this paper we have studied the transfer matrix spectrum of the class of
cyclic 6-vertex representations of the reflection algebra in the case of one
completely general and one triangular boundary matrices and for bulk
parameters satisfying some specific constraints. Our result is the complete
characterization of the spectrum (eigenvalues and eigenstates) of this class
of models both by a discrete system of Baxter's like second order difference
equations and by a single inhomogeneous TQ-functional equation within a
class of polynomial Q-functions.

The present paper represents a natural starting point to solve the spectral
problem in the most general setting. In order to do so we need to generalize
to the 6-vertex cyclic representations the Baxter's gauge transformations
used in the 6-vertex spin 1/2 highest weight representations and prove then
the pseudo-diagonalizability of the associated family of gauge transformed $%
\mathcal{B_-}$-operators. These points are currently under analysis.

An interesting point to which we would like to come back in future
investigations is the explicit construction of the Q-operator for the cyclic
representations of the 6-vertex reflection algebra. This can lead to new
connections with transfer matrices of exactly solvable class of models of
non 6-vertex type. Indeed, let us recall that for the special class of
6-vertex cyclic representations of the Yang-Baxter algebra, parametrized by
points on the algebraic curve (\ref{chPFaxVcurve-eq}), the associated
integrable models have some remarkable connections with the inhomogeneous
p-state chiral Potts models. Indeed, the chiral Potts transfer matrices play
the role of the Q-operators for the transfer matrix associated to the cyclic
representations of the 6-vertex Yang-Baxter algebra \cite{OpenCyBS90}.

\section*{Acknowledgements}
J. M. M. and G. N. are supported by CNRS and ENS Lyon;  B. P. is supported by ENS Lyon and ENS Cachan.


\appendix

\section{Appendix}

\subsection{General proof of diagonalizability of $\mathcal{B}_{-}(\protect%
\lambda )$}

The generator $\mathcal{B}_{-}(\lambda )$ admits the following boundary bulk
decomposition:%
\begin{equation}
\mathcal{B}_{-}(\lambda )=-a_{-}(\lambda )A(\lambda )B(1/\lambda
)+b_{-}(\lambda )A(\lambda )A(1/\lambda )-c_{-}(\lambda )B(\lambda
)B(1/\lambda )+d_{-}(\lambda )B(\lambda )A(1/\lambda ),
\end{equation}%
if the inner boundary matrix is non-diagonal and the bulk parameters are
generals, i.e. if it holds:%
\begin{equation}
\text{\textsc{b}}_{-}=\left( -1\right) ^{\mathsf{N}}\kappa _{-}e^{\tau
_{-}}\prod_{a=1}^{\mathsf{N}}\alpha _{n}\beta _{n}\neq 0,
\label{B-non-nilpotent}
\end{equation}%
then it trivially follows that $\mathcal{B}_{-}(\lambda )$ admits the
following functional form:%
\begin{equation}
\mathcal{B}_{-}(\lambda )=\text{\textsc{b}}_{-}(\frac{\lambda ^{2}}{q}-\frac{%
q}{\lambda ^{2}})\prod_{a=1}^{\mathsf{N}}(\frac{\lambda }{\mathcal{B}_{-,a}}-%
\frac{\mathcal{B}_{-,a}}{\lambda })(\lambda \mathcal{B}_{-,a}-\frac{1}{%
\lambda \mathcal{B}_{-,a}}),
\end{equation}%
where the $\mathcal{B}_{-,a}$ are invertible commuting operators. The above
functional form implies that under the condition $\left( \ref%
{B-non-nilpotent}\right) $ the operator family is not nilpotent, in
particular does not exist any state annihilated by $\mathcal{B}_{-}(\lambda
) $ for any $\lambda \in \mathbb{C}$. That is does not exist a reference
state and we cannot use ABA to analyze the spectral problem of the transfer
matrix. Here, we show that under the condition $\left( \ref{B-non-nilpotent}%
\right) $ for general values of the bulk parameters the $\mathcal{B}%
_{-}(\lambda )$ is indeed diagonalizable and with simple spectrum. Let us
first prove the following lemma:

\begin{lemma}
There exists at least one simultaneous eigenstate:%
\begin{equation}
|\Omega _{R}\rangle ,\text{ \ }\langle \Omega_{L}|
\end{equation}%
of the one parameter family of commuting operators $\mathcal{B}_{-}(\lambda
) $:%
\begin{align}
& \mathcal{B}_{-}(\lambda )|\Omega _{R}\rangle =|\Omega _{R}\rangle \text{%
\textsc{b}}_{-}(\frac{\lambda ^{2}}{q}-\frac{q}{\lambda ^{2}})\prod_{a=1}^{%
\mathsf{N}}(\frac{\lambda }{\bar{b}_{-,a}}-\frac{\bar{b}_{-,a}}{\lambda }%
)(\lambda \bar{b}_{-,a}-\frac{1}{\lambda \bar{b}_{-,a}}) \\
& \langle \Omega _{L}|\mathcal{B}_{-}(\lambda )=\text{\textsc{b}}_{-}(\frac{%
\lambda ^{2}}{q}-\frac{q}{\lambda ^{2}})\prod_{a=1}^{\mathsf{N}}(\frac{%
\lambda }{\hat{b}_{-,a}}-\frac{\hat{b}_{-,a}}{\lambda })(\lambda \hat{b}%
_{-,a}-\frac{1}{\lambda \hat{b}_{-,a}})\langle \Omega _{L}|.
\end{align}
\end{lemma}

\begin{proof}
We can always put $\mathcal{B}_{-,1}$ in a Jordan normal form, let us denote
with B$_{1}$ the right eigenspace associated to a given eigenvalue $\bar{b}%
_{-,1}\neq 0$ of $\mathcal{B}_{-,1}$, as $\mathcal{B}_{-,1}$ is invertible.
If this eigenspace is one dimensional we have found our simultaneous
eigenstate of $\mathcal{B}_{-}(\lambda )$. If this is not the case then by
the commutativity we have that B$_{1}$ is an invariant space w.r.t. $%
\mathcal{B}_{-,n}$\ for any\ $n\in \left\{ 1,...,N\right\} $. So we can
always put $\mathcal{B}_{-,2}$ in a Jordan normal form in B$_{1}$, let us
denote with B$_{1,2}$ the eigenspace associated to a given eigenvalue $\bar{b%
}_{-,2}\neq 0$ of $\mathcal{B}_{-,2}$. If this eigenspace is one dimensional
we have found our simultaneous eigenstate of $\mathcal{B}_{-}(\lambda )$
otherwise we can reiterate this procedure. We can have two possibilities
both at the step $n\leq N$ we find that the eigenspace B$_{1,...,n}$ is
one-dimensional or we arrive up to the eigenspace B$_{1,...,N}$ in both the
cases we have (at least) one simultaneous eigenstate of the $\mathcal{B}%
_{-,n}$\ for any\ $n\in \left\{ 1,...,N\right\} $ and so one eigenstate of $%
\mathcal{B}_{-}(\lambda )$. Similarly, we can prove the existence of $%
\langle \Omega _{L}|$.
\end{proof}

From the previous Lemma and the reflection algebra equation we can prove the
following proposition.

\begin{proposition}
Under the condition $\left( \ref{B-non-nilpotent}\right) $ for almost all
the values of the bulk parameters, the operator family $\mathcal{B}%
_{-}(\lambda )$ is diagonalizable and it has simple spectrum and its average
value is central and it holds:%
\begin{equation}
\mathbb{B}_{-}(\lambda )=\prod_{a=1}^{p}\mathcal{B}_{-}(\lambda q^{a})=\text{%
\textsc{b}}_{-}^{p}(\frac{\lambda ^{2p}}{1}-\frac{1}{\lambda ^{2p}}%
)\prod_{a=1}^{\mathsf{N}}(\frac{\lambda ^{p}}{\bar{b}_{-,a}^{p}}-\frac{\bar{b%
}_{-,a}^{p}}{\lambda ^{p}})(\lambda ^{p}\bar{b}_{-,a}^{p}-\frac{1}{\lambda
^{p}\bar{b}_{-,a}^{p}}),
\end{equation}%
with:%
\begin{equation}
\bar{b}_{-,m}^{p}\neq \bar{b}_{-,n}^{p},\text{ }\forall n\neq m\in \left\{
1,...,N\right\} .  \label{Simplicity}
\end{equation}
\end{proposition}

\begin{proof}
Let us observe that by definition the operator family $\mathcal{B}%
_{-}(\lambda )$ is a polynomial in the bulk parameters so the same must be
true for its spectrum. This implies in particular that defined%
\begin{equation}
\hat{b}_{n,m}=\bar{b}_{-,n}^{p}-\bar{b}_{-,m}^{p},\text{ }\forall n\neq m\in
\left\{ 1,...,N\right\} ,
\end{equation}%
or they are identically zero or they can be zero only over subspaces of
nonzero codimension in the space of the bulk parameters. So here we have
just to prove that the $\hat{b}_{n,m}$ are not identically zero to derive
that $\left( \ref{Simplicity}\right) $ holds for almost all the values of
the parameters. To do so we can just recall that from the results derived in
the previous section it holds:%
\begin{equation}
\bar{b}_{-,m}^{p}=q^{p/2}\mu _{m,+}^{p},\text{ }\forall m\in \left\{
1,...,N\right\} ,
\end{equation}%
under the condition $B(\lambda )$ nilpotent, i.e. $a_{n}^{p}=-b_{n}^{p}$\
for any $n\in \left\{ 1,...,N\right\} $. From which the condition $\left( %
\ref{Simplicity}\right) $ holds as soon as we impose:%
\begin{equation}
\beta _{n}^{p}/\alpha _{n}^{p}\neq \beta _{m}^{p}/\alpha _{m}^{p},\text{ }%
\forall n\neq m\in \left\{ 1,...,N\right\} .
\end{equation}%
Once we have proven this statement, then we have just to use the reflection
algebra to construct an eigenbasis of $\mathcal{B}_{-}(\lambda )$ and this
is done just by a repeated action of the generators $\mathcal{A}_{-}(\lambda
)$ computed in the zeros of $\mathcal{B}_{-}(\lambda )$ on the eigenstates $%
|\Omega _{R}\rangle $ and $\langle\Omega _{L}|$. That is we repeat the
construction of the eigenbasis presented in Section \ref{SoV-Rep-Refle-Alge}
by substituting to the reference states defined in $\left( \ref%
{Ref-state-Con-L/R}\right) $ with the $\mathcal{B}_{-}(\lambda )$ on the
eigenstates $|\Omega _{R}\rangle $ the state $\langle \Omega _{L}|$. Note
that such an action can generate a null vector only if some of the zeros of $%
\mathcal{B}_{-}(\lambda )$ coincides with the zeros of the quantum
determinant anyhow as discussed in Section \ref{SoV-Rep-Refle-Alge} under
the condition $\left( \ref{condition-SoV-2}\right) $ we are always able to
chose an appropriate set of $p^{N}$ zeros of $\mathcal{B}_{-}(\lambda )$
which do not coincides with those of the quantum determinant so that we
generate exactly $p^{N}$ independent states.
\end{proof}

\subsection{The lattice sine-Gordon model with integrable boundaries}

In this appendix we point out that the results on the transfer matrix
spectrum for general cyclic representations developed in this paper apply
also to characterize the spectrum of the transfer matrix of the lattice
sine-Gordon model with integrable open boundary conditions. Let us recall
that the Lax operator defining the lattice sine-Gordon model has the
following form: 
\begin{equation}
L_{a,n}^{sG}(\lambda |\text{\textsc{u}}_{n},\text{\textsc{v}}_{n},\kappa
_{n},r_{n},s_{n})\equiv \left( 
\begin{array}{cc}
\text{\textsc{u}}_{n}\left( q^{-1/2}\kappa _{n}^{2}r_{n}s_{n}\text{\textsc{v}%
}_{n}+\frac{q^{1/2}s_{n}}{\text{\textsc{v}}_{n}r_{n}}\right) & \frac{\kappa
_{n}}{i}\left( \lambda r_{n}\text{\textsc{v}}_{n}-\frac{1}{\lambda r_{n}%
\text{\textsc{v}}_{n}}\right) \\ 
\frac{\kappa _{n}}{i}\left( \frac{\lambda }{r_{n}\text{\textsc{v}}_{n}}-%
\frac{r_{n}\text{\textsc{v}}_{n}}{\lambda }\right) & \text{\textsc{u}}%
_{n}^{-1}\left( \frac{q^{1/2}r_{n}\text{\textsc{v}}_{n}}{s_{n}}+\frac{%
q^{-1/2}\kappa _{n}^{2}}{s_{n}r_{n}\text{\textsc{v}}_{n}}\right)%
\end{array}%
\right) _{a},
\end{equation}%
each local representation being defined as the representation of a local
Weyl algebra%
\begin{equation}
\text{\textsc{u}}_{n}\text{\textsc{v}}_{m}=q^{\delta _{n,m}}\text{\textsc{v}}%
_{m}\text{\textsc{u}}_{n}\text{ \ \ }\forall n,m\in \{1,...,\mathsf{N}\},
\end{equation}%
associated to a root of unit $q$, where \textsc{u}$_{n}$ and \textsc{v}$_{n}$
are the Weyl algebra generators on the Hilbert space $\mathcal{R}_{n}$. Let
us introduce the monodromy matrices for the Yang-Baxter and reflection
equations of the lattice sine-Gordon model, they read:%
\begin{eqnarray}
M_{a}^{sG}(\lambda ) &=&L_{a,\mathsf{N}}(\lambda q^{-1/2}/\xi _{\mathsf{N}%
})\cdots L_{a,1}(\lambda q^{-1/2}/\xi _{1})\in \text{End}(\mathbb{C}%
^{2}\otimes \mathcal{H}), \\
\mathcal{U}_{a,-}^{sG}(\lambda ) &=&M_{a}^{sG}(\lambda )K_{a,-}(\lambda )%
\hat{M}_{a}^{sG}(\lambda )\in \text{End}(\mathbb{C}^{2}\otimes \mathcal{H}),
\end{eqnarray}%
and also the boundary transfer matrix of the sine-Gordon model, which reads:%
\begin{equation}
\mathcal{T}^{sG}(\lambda )\equiv tr_{a}\{K_{a,+}(\lambda )\mathcal{U}%
_{a,-}^{sG}(\lambda )\}.
\end{equation}%
Let us remark that we have used the upper index sG in the above two
monodromy matrices and transfer matrix to point out that they are related to
the sine-Gordon model while we will continue to denote with $M_{a}(\lambda )$%
, $\mathcal{U}_{a,-}(\lambda )$ and $\mathcal{T}(\lambda )$ those associated
to the original $\tau _{2}$-model. The following lemma establishes the
connection between the monodromy and transfer matrices of the sine-Gordon
model and the original $\tau _{2}$-model.

\begin{lemma}
Let us denote $\mathsf{N}=2\mathsf{M}+\mathsf{x}$ with $\mathsf{x}\in
\left\{ 0,1\right\} $ and let us impose the following identification of the
generators of the local Weyl algebras:%
\begin{equation}
\text{\textsc{u}}_{2n+\mathsf{y}}=u_{2n+\mathsf{y}}^{(1-2\mathsf{x})(1-2%
\mathsf{y})},\text{ \ \textsc{v}}_{2n+\mathsf{y}}=v_{2n+\mathsf{y}}^{(1-2%
\mathsf{x})(1-2\mathsf{y})},  \label{local-op-id-sG-t2}
\end{equation}%
with $\mathsf{y}\in \left\{ 0,1\right\} $ and\ $2n+\mathsf{y}\in \left\{
1,...,\mathsf{N}\right\} $. Then the following identity holds:%
\begin{equation}
M_{a}^{sG}(\lambda )=M_{a}(\lambda )\left( \sigma _{a}^{x}\right) ^{\mathsf{x%
}},  \label{YB-monodromy-sG-t2}
\end{equation}%
once we define the parameters of the $\tau _{2}$-model in terms of those of
the lattice sine-Gordon model as it follows:%
\begin{eqnarray}
a_{2n+\mathsf{y}} &=&\kappa _{2n+\mathsf{y}}^{2}\left( r_{2n+\mathsf{y}%
}s_{2n+\mathsf{y}}\right) ^{(1-2\mathsf{x})(1-2\mathsf{y})},\text{ \ \ }%
b_{2n+\mathsf{y}}=\left( \frac{s_{2n+\mathsf{y}}}{r_{2n+\mathsf{y}}}\right)
^{(1-2\mathsf{x})(1-2\mathsf{y})},\text{ \ \ }  \label{Par-1} \\
c_{2n+\mathsf{y}} &=&\left( \frac{r_{2n+\mathsf{y}}}{s_{2n+\mathsf{y}}}%
\right) ^{(1-2\mathsf{x})(1-2\mathsf{y})},\text{ \ \ \ \ \ }d_{2n+\mathsf{y}%
}=\frac{\kappa _{2n+\mathsf{y}}^{2}}{\left( r_{2n+\mathsf{y}}s_{2n+\mathsf{y}%
}\right) ^{(1-2\mathsf{x})(1-2\mathsf{y})}},\text{ \ } \\
\alpha _{2n+\mathsf{y}} &=&\frac{\kappa _{2n+\mathsf{y}}r_{2n+\mathsf{y}%
}^{(1-2\mathsf{x})(1-2\mathsf{y})}}{i\xi _{2n+\mathsf{y}}},\text{ \ \ \ \ \ }%
\beta _{2n+\mathsf{y}}=\frac{\kappa _{2n+\mathsf{y}}\xi _{2n+\mathsf{y}}}{%
ir_{2n+\mathsf{y}}^{(1-2\mathsf{x})(1-2\mathsf{y})}}.  \label{Par-3}
\end{eqnarray}%
Moreover, under the above conditions and imposing the following
identifications:%
\begin{equation}
\tau _{\epsilon }=\epsilon ^{\mathsf{x}}\tau _{\epsilon }^{sG},\text{ }%
\kappa _{\epsilon }=\epsilon ^{\mathsf{x}}\kappa _{\epsilon }^{sG},\text{ }%
\zeta _{\epsilon }=\left( \zeta _{\epsilon }^{sG}\right) ^{\epsilon ^{%
\mathsf{x}}}\text{ for }\epsilon =\pm ,  \label{boundary-par-sG-t2}
\end{equation}%
\newline
of the boundary parameters of the $\tau _{2}$-model and the lattice
sine-Gordon model, we get:%
\begin{equation}
\mathcal{U}_{a,-}(\lambda )=\mathcal{U}_{a,-}^{sG}(\lambda ),\text{ \ }%
\mathcal{T}(\lambda )=\mathcal{T}^{sG}(\lambda ).
\label{Baundary-identities-sG-t2}
\end{equation}
\end{lemma}

\begin{proof}
It is simple to observe that the following identities hold:%
\begin{equation}
L_{a,n}^{sG}(\lambda /\xi _{n}|\text{\textsc{u}}_{n},\text{\textsc{v}}%
_{n},\kappa _{n},r_{n},s_{n})=L_{a,n}(\lambda )\sigma _{a}^{x}
\label{YB-Lax-sG-t2}
\end{equation}%
if:%
\begin{equation}
\text{\textsc{u}}_{n}=u_{n},\text{ \ \textsc{v}}_{n}=v_{n},
\end{equation}%
and%
\begin{eqnarray}
a_{n} &=&\kappa _{n}^{2}r_{n}s_{n},\text{ \ }b_{n}=s_{n}/r_{n},\text{ \ }%
c_{n}=r_{n}/s_{n},  \label{sG-t2-par1} \\
d_{n} &=&\kappa _{n}^{2}/(r_{n}s_{n}),\text{ \ }\alpha _{n}=-i\kappa
_{n}r_{n}/\xi _{n},\text{ \ }\beta _{n}=-i\kappa _{n}\xi _{n}/r_{n}.
\label{sG-t2-par2}
\end{eqnarray}%
Similarly by direct computations it is easy to show that defined:%
\begin{equation}
\tilde{L}_{a,n}^{sG}(\lambda )\equiv \sigma _{a}^{x}L_{a,n}^{sG}(\lambda |%
\text{\textsc{u}}_{n},\text{\textsc{v}}_{n},\kappa _{n},r_{n},s_{n})\sigma
_{a}^{x}
\end{equation}%
it holds:%
\begin{equation}
\tilde{L}_{a,n}^{sG}(\lambda )=L_{a,n}^{sG}(\lambda |\text{\textsc{u}}%
_{n}^{-1},\text{\textsc{v}}_{n}^{-1},\kappa _{n},r_{n}^{-1},s_{n}^{-1}).
\label{YB-Lax-sG-to-sG}
\end{equation}%
Let us now observe that for $\mathsf{x}=0$, we can write:%
\begin{equation}
M_{a}^{sG}(\lambda )=[L_{a,2\mathsf{M}}(\frac{\lambda q^{-1/2}}{\xi _{2%
\mathsf{M}}})\sigma _{a}^{x}][\tilde{L}_{a,2\mathsf{M}-1}(\frac{\lambda
q^{-1/2}}{\xi _{2\mathsf{M}-1}})\sigma _{a}^{x}]\cdots \lbrack L_{a,2}(\frac{%
\lambda q^{-1/2}}{\xi _{2}})\sigma _{a}^{x}][\tilde{L}_{a,1}(\frac{\lambda
q^{-1/2}}{\xi _{1}})\sigma _{a}^{x}]
\end{equation}%
while for $\mathsf{x}=1$, we can write:%
\begin{equation}
M_{a}^{sG}(\lambda )\sigma _{a}^{x}=[L_{a,2\mathsf{M}+1}(\frac{\lambda
q^{-1/2}}{\xi _{2\mathsf{M}+1}})\sigma _{a}^{x}][\tilde{L}_{a,2\mathsf{M}}(%
\frac{\lambda q^{-1/2}}{\xi _{2\mathsf{M}}})\sigma _{a}^{x}]\cdots \lbrack 
\tilde{L}_{a,2}(\frac{\lambda q^{-1/2}}{\xi _{2}})\sigma _{a}^{x}][L_{a,1}(%
\frac{\lambda q^{-1/2}}{\xi _{1}})\sigma _{a}^{x}],
\end{equation}%
then these two identities together with the identity $\left( \ref%
{YB-Lax-sG-t2}\right) $, $\left( \ref{YB-Lax-sG-to-sG}\right) $ and the
parametrization $\left( \ref{sG-t2-par1}\right) $-$\left( \ref{sG-t2-par2}%
\right) $ imply the identity $\left( \ref{YB-monodromy-sG-t2}\right) $ with
the parametrization $\left( \ref{Par-1}\right) $-$\left( \ref{Par-3}\right) $%
and the local operator identification $\left( \ref{local-op-id-sG-t2}\right) 
$. Finally, let us observe that from the identity $\left( \ref%
{YB-monodromy-sG-t2}\right) $ it follows:%
\begin{equation}
\hat{M}_{a}^{sG}(\lambda )\equiv (-1)^{\mathsf{N}}\,\sigma _{a}^{y}\,\left(
M_{a}^{sG}(1/\lambda )\right) ^{t_{0}}\,\sigma _{a}^{y}=\left( -\sigma
_{a}^{x}\right) ^{\mathsf{x}}\hat{M}_{a}(\lambda ),
\end{equation}%
which together with:%
\begin{equation}
K_{a,+}(\lambda |\tau _{\epsilon },\text{ }\kappa _{\epsilon },\text{ }\zeta
_{\epsilon })=\left( \sigma _{a}^{x}\right) ^{\mathsf{x}}K_{a,+}^{sG}(%
\lambda |\tau _{-}^{sG},\kappa _{-}^{sG},\zeta _{-}^{sG})\left( -\sigma
_{a}^{x}\right) ^{\mathsf{x}},
\end{equation}%
holding under the parametrization $\left( \ref{boundary-par-sG-t2}\right) $,
implies the identities $\left( \ref{Baundary-identities-sG-t2}\right) $.
\end{proof}

\subsection{Reduction to inhomogeneous chiral Potts representations}

In this appendix we want to point out that a nontrivial class of
representations of the inhomogeneous chiral Potts model can be described on
the closed chain in the space of the parameters of the $\tau _{2}$-model
considered in this paper. In order to do so let us recall that the transfer
matrix $\mathsf{T}_{\lambda }^{{\small \text{chP}}}$ of the inhomogeneous
chiral Potts model \cite{OpenCyBS90} is characterized by the following
kernel:%
\begin{equation}
\mathsf{T}_{\lambda }^{{\small \text{chP}}}(\text{z},\text{z}^{\prime
})\equiv \langle \text{z}|\mathsf{T}_{\lambda }^{{\small \text{chP}}}|\text{z%
}^{\prime }\rangle =\prod_{n=1}^{\mathsf{N}}W_{\text{q}_{n}\text{p}%
}(z_{n}/z_{n}^{\prime })\bar{W}_{\text{r}_{n}\text{p}}(z_{n}/z_{n+1}^{\prime
}),  \label{chPFaxVkernel}
\end{equation}%
in the basis $\langle $z$|\equiv \langle z_{1},...,z_{N}|$ and $|$z$\rangle
\equiv |z_{1},...,z_{N}\rangle $ formed by the left and right $u_{n}$%
-eigenstates:%
\begin{equation}
\langle \text{z}|u_{n}=z_{n}\langle \text{z}|,\text{ \ \ \ }u_{n}|\text{z}%
\rangle =z_{n}|\text{z}\rangle \text{ with }z_{n}\in \mathbb{S}_{p}\equiv
\{q^{2r},\text{ }r=1,..,p\},
\end{equation}%
and where:%
\begin{equation}
\lambda =t_{\text{p}}^{-1/2}\mathsf{c}_{0},\text{ \ \ p, r}_{n}\text{, q}%
_{n}\in \mathcal{C}_{k},\mathsf{c}_{0}\in \mathbb{C}\text{.}
\end{equation}%
The algebraic curve $\mathcal{C}_{k}$ of modulus $k$ is by definition the
locus of the points p $\equiv (a_{\text{p}},b_{\text{p}},c_{\text{p}},d_{%
\text{p}})\in \mathbb{C}^{4}$ which satisfy the equations: 
\begin{equation}
x_{\text{p}}^{p}+y_{\text{p}}^{p}=k(1+x_{\text{p}}^{p}y_{\text{p}}^{p}),%
\text{ \ \ }kx_{\text{p}}^{p}=1-k^{^{\prime }}s_{\text{p}}^{-p},\text{ \ \ }%
ky_{\text{p}}^{p}=1-k^{^{\prime }}s_{\text{p}}^{p},  \label{chPFaxVcurve-eq}
\end{equation}%
where:%
\begin{equation}
x_{\text{p}}\equiv a_{\text{p}}/d_{\text{p}},\text{ \ }y_{\text{p}}\equiv b_{%
\text{p}}/c_{\text{p}},\text{ \ }s_{\text{p}}\equiv d_{\text{p}}/c_{\text{p}%
},\,t_{\text{p}}\equiv x_{\text{p}}y_{\text{p}},\text{ \ }k^{2}+(k^{^{\prime
}})^{2}=1,
\end{equation}%
and $W_{\text{qp}}(z(n))$ and $\bar{W}_{\text{qp}}(z(n))$ are the Boltzmann
weights of the chiral Potts model. Then, \textsf{T}$_{\lambda }^{{\small 
\text{chP}}}$ is a Baxter $\mathsf{Q}$-operator w.r.t. the bulk $\tau _{2}$%
-transfer matrix in $\mathcal{H}_{\mathsf{N}}$.

Let us here directly characterize the class of the inhomogeneous chiral
Potts representations once we restrict the space of the parameters to that
used in the section \ref{functional equation}; in particular, we assume that
it holds:%
\begin{equation}
b_{n}=-q^{-1}a_{n},\text{ \ }d_{n}=-q^{-1}c_{n}.  \label{restriction1-chP}
\end{equation}%
The parameters of the $\tau _{2}$-Lax operators are written in terms of the
coordinate of the points p, r$_{n}$, q$_{n}$ by using the equations (5.3) of
the paper \cite{OpenCyGN12}. Then we have that the points r$_{n}$, q$_{n}$
are elements of $\mathcal{C}_{k}$ if and only if beyond $\left( \ref%
{restriction1-chP}\right) $ the parameters of the $\tau _{2}$-Lax operators
satisfy the following conditions:%
\begin{equation}
\alpha _{n}\beta _{n}=a_{n}c_{n}
\end{equation}%
and%
\begin{equation}
\left( \frac{\mathsf{c}_{0}\alpha _{n}}{q^{1/2}a_{n}}\right) ^{p}+\left( 
\frac{q^{1/2}\mathsf{c}_{0}\alpha _{n}}{c_{n}}\right) ^{p}=k\left( 1+\left( 
\frac{\mathsf{c}_{0}^{2}\alpha _{n}^{2}}{c_{n}a_{n}}\right) ^{p}\right) .
\end{equation}%
Under these constraints the class of the inhomogeneous chiral Potts
representations is characterized by the following identity:%
\begin{equation}
\text{r}_{n}=\Delta (\text{q}_{n}),\text{ \ \ }\forall n\in \{1,...,\mathsf{N%
}\},
\end{equation}%
where the q$_{n}$ are free elements of $\mathcal{C}_{k}$ and $\Delta $ is
the following discrete automorphism of the curve:%
\begin{equation}
\Delta :\text{x}=(a_{\text{x}},b_{\text{x}},c_{\text{x}},d_{\text{x}})\in 
\mathcal{C}_{k}\rightarrow \Delta (\text{x})=(b_{\text{x}},a_{\text{x}},d_{%
\text{x}},c_{\text{x}})\in \mathcal{C}_{k},
\end{equation}%
which implies:%
\begin{equation}
x_{\text{p}_{n}}=y_{\text{q}_{n}},\text{ }y_{\text{p}_{n}}=x_{\text{q}_{n}},%
\text{ }s_{\text{p}_{n}}=s_{\text{q}_{n}}^{-1}.
\end{equation}%
Finally, this class of representations reduce to the superintegrable chiral
Potts model under the following special homogeneous limits:%
\begin{equation}
x_{\text{q}_{n}}^{p}\rightarrow \left( 1+k^{\prime }\right) /k,\text{ \ \ }%
\forall n\in \{1,...,\mathsf{N}\}.
\end{equation}

\subsection{Reduction to the XXZ spin $s$ open chains at the $p=2s+1$ roots
of unit}

Here we show that imposing a set of conditions on the parameters of the $%
\tau _{2}$-Lax operator we can reduce it to the one of the spin $s=\left(
p-1\right) /2$ XXZ case at the $p$ roots of unit. This has the interesting
consequence that the analysis done of the open $\tau _{2}$-chain reduces for
these special representations to that of an open spin chain under the same
boundary conditions. In particular, we have that our functional equation
characterization of the spectrum under a special homogeneous limit defines
the spectrum of the following local Hamiltonian given by fusion and the
Sklyanin formula:%
\begin{equation}
\mathcal{H}_{s}=c_{0}\frac{d}{d\lambda }\mathcal{T}_{p}(\lambda )_{\,\vrule %
height13ptdepth1pt\>{\lambda =q}^{s}}+\text{constant,}
\end{equation}%
with%
\begin{equation}
\mathcal{H}_{s}=\sum_{n=1}^{\mathsf{N}-1}H_{n,n+1}^{\left( s\right) }+\frac{d%
}{d\lambda }K_{1,-}^{(p)}(q^{s})+\frac{\text{tr}_{\mathsf{0}}\{K_{\mathsf{0}%
,+}^{(p)}(q^{s})H_{\mathsf{0},\mathsf{N}}^{\left( s\right) }\}}{\text{tr}_{%
\mathsf{0}}\{K_{\mathsf{0},+}^{(p)}(q^{s})\}},
\end{equation}%
where $\mathcal{T}_{p}(\lambda )$ is the $p$-fused open transfer matrix, $%
H_{n,n+1}^{\left( s\right) }$ is the two sites local Hamiltonian of the spin 
$s$ XXZ chain, $K_{\mathsf{0},\pm }^{(p)}(\lambda )$ are the $p\times p$
matrices $p$-fused scalar solutions of the reflection algebra obtained by
doing the fusion $p-1$ times starting from the original $2\times 2$ scalar
solutions $K_{0,\pm }(\lambda )$, respectively.

\begin{lemma}
Let us fix the parameters of the $\tau _{2}$-representations as follows:%
\begin{eqnarray}
\alpha _{n} &=&\beta _{n}=1/2,\text{ }a_{n}=q^{-1/2}/2i, \\
b_{n} &=&iq^{1/2}/2,\text{ }c_{n}=q^{-1/2}/2i,\text{ }d_{n}=iq^{1/2}/2,
\end{eqnarray}%
and let%
\begin{equation}
L_{a,n}^{XXZ}(\lambda )=\left( 
\begin{array}{cc}
\left[ \lambda q^{s+S_{n}^{z}/2}-1/(\lambda q^{s+S_{n}^{z}/2})\right] /2 & 
S_{n}^{-} \\ 
S_{n}^{+} & \left[ \lambda q^{s-S_{n}^{z}/2}-1/(\lambda q^{s-S_{n}^{z}/2})%
\right] /2%
\end{array}%
\right)
\end{equation}%
be the Lax operator of the spin $s$ XXZ chain with anisotropy cosh$\,\eta
=\left( q+1/q\right) /2$, then it holds:%
\begin{equation}
L_{a,n}^{XXZ}(\lambda )=L_{a,n}(\lambda /q)
\end{equation}%
for $s=\left( p-1\right) /2$ and $q=e^{i\pi p^{\prime }/p}$ with $p^{\prime
} $ even and $p$ odd and coprime, which is equivalent to the following
identities among the generators of the local algebras:%
\begin{equation}
S_{n}^{+}=u_{n}^{-1}\left( v_{n}-1/v_{n}\right) /2i,\text{ \ }%
S_{n}^{-}=u_{n}\left( v_{n}/q-q/v_{n}\right) /2i,
\label{S_+-u-v-identification}
\end{equation}%
and%
\begin{equation}
S_{n}^{z}=\frac{2p}{i\pi p^{\prime }}\log v_{n}-(p+1)\in
\{-2s,-2(s-1),...,2s\}\text{ \ mod\thinspace }2p
\label{S_z-v-identification}
\end{equation}
\end{lemma}

\begin{proof}
Let us denote 
\begin{equation}
\overline{\left\vert a,n\right\rangle }=\left( 
\begin{array}{ccccc}
0 & \cdots & 1 & \cdots & 0%
\end{array}%
\right) ^{t_{0}},\text{ }a\in \{1,...,2s+1\}
\end{equation}%
the element $a$ of the canonical basis given by the column vector with all
elements zero except that in row $a$ which is 1. In this basis we have the
following representations for the generators%
\begin{equation}
S_{n}^{+}=\left( 
\begin{array}{cccc}
0 & f(1) &  &  \\ 
& \ddots &  &  \\ 
&  & \ddots & f(2s) \\ 
&  &  & 0%
\end{array}%
\right) ,\text{ \ \ }S_{n}^{-}=\left( 
\begin{array}{cccc}
0 &  &  &  \\ 
f(1) & \ddots &  &  \\ 
& \ddots & \ddots &  \\ 
&  & f(2s) & 0%
\end{array}%
\right) ,
\end{equation}%
where:%
\begin{equation}
f(j)=\sqrt{\sinh j\eta \sinh (p-j)\eta }=i(q^{j}-q^{-j})/2,
\end{equation}%
and the second identity holds for $q^{p}=1$, and:%
\begin{equation}
S_{n}^{z}=\left( 
\begin{array}{cccc}
2s & 0 &  &  \\ 
0 & \ddots & \ddots &  \\ 
& \ddots & \ddots & 0 \\ 
&  & 0 & -2s%
\end{array}%
\right) .
\end{equation}

Let us now impose that in our representation the $v_{n}$-eigenstates
coincide with the elements of the canonical basis:%
\begin{equation}
\left\vert p-a,n\right\rangle =\overline{\left\vert a+1,n\right\rangle }%
\text{ \ }\forall a\in \{0,...,p-1\},
\end{equation}%
we can verify now the formulae $\left( \ref{S_+-u-v-identification}\right) $-%
$\left( \ref{S_z-v-identification}\right) $. The formula $\left( \ref%
{S_z-v-identification}\right) $ is equivalent to:%
\begin{equation}
v_{n}=q^{(S_{n}^{z}+p+1)/2}
\end{equation}%
which holds for the following identities:%
\begin{equation}
q^{(S_{n}^{z}+p+1)/2}\overline{\left\vert a+1,n\right\rangle }=\overline{%
\left\vert a+1,n\right\rangle }q^{(2(s-a)+p+1)/2}=\left\vert
p-a,n\right\rangle q^{p-a}=v_{n}\left\vert p-a,n\right\rangle .
\end{equation}%
Similarly we have:%
\begin{eqnarray}
S_{n}^{+}\overline{\left\vert a,n\right\rangle } &=&\overline{\left\vert
a+1,n\right\rangle }f(a)=\overline{\left\vert a+1,n\right\rangle }%
(q^{-a}-q^{a})/2i \\
&=&\left[ u_{n}^{-1}\left( v_{n}-1/v_{n}\right) /2i\right] \left\vert
p-a,n\right\rangle \text{ \ }\forall a\in \{1,...,p\}
\end{eqnarray}%
and%
\begin{eqnarray}
S_{n}^{-}\overline{\left\vert a,n\right\rangle } &=&\overline{\left\vert
a-1,n\right\rangle }f(a-1)=\overline{\left\vert a-1,n\right\rangle }%
(q^{1-a}-q^{a-1})/2i \\
&=&\left[ u_{n}\left( v_{n}/q-q/v_{n}\right) /2i\right] \left\vert
p+2-a,n\right\rangle \text{ \ }\forall a\in \{1,...,p\}.
\end{eqnarray}
\end{proof}

\end{document}